\documentclass[a4paper,11pt,DIV=12,abstract=true,bibliography=totoc]{scrartcl}
\usepackage[T1]{fontenc}
\usepackage[utf8]{inputenc}
\usepackage[english]{babel}
\usepackage{amssymb,amsfonts,amsmath,amsthm}
\usepackage{mathtools}
\usepackage[hidelinks]{hyperref}
\usepackage{authblk}
\usepackage{booktabs}
\usepackage{subcaption}
\usepackage{enumitem}
\usepackage{bbm}
\usepackage{pgfplots}
\usepackage{lmodern}

\mathtoolsset{showonlyrefs}
\pgfplotsset{compat=newest}
\newlength\figureheight
\newlength\figurewidth
\usepgfplotslibrary{external}
\newtheorem{theorem}{Theorem}[section]
\newtheorem{lemma}[theorem]{Lemma}
\newtheorem{proposition}[theorem]{Proposition}
\newtheorem{corollary}[theorem]{Corollary}
\theoremstyle{definition}
\newtheorem{definition}[theorem]{Definition}
\newtheorem{example}[theorem]{Example}
\newtheorem{remark}[theorem]{Remark}

\numberwithin{equation}{section}
\numberwithin{table}{section}
\numberwithin{figure}{section}

\newcommand{\R}{\mathbb{R}}

\newcommand{\transp}{\top}
\newcommand{\rmd}{\mathrm{d}}
\newcommand{\rme}{\mathrm{e}}
\newcommand{\ones}{\mathbf{1}_d}
\newcommand{\calA}{\mathcal{A}}
\newcommand{\calF}{\mathcal{F}}
\newcommand{\calG}{\mathcal{G}}
\newcommand{\calK}{\mathcal{K}}

\DeclareMathOperator{\E}{\mathbb{E}}

\DeclareMathOperator{\cov}{cov}

\DeclareMathOperator{\COA}{COA}
\DeclareMathOperator{\RDR}{RDR}

\begin{document}

\title{Robust Utility Maximizing Strategies under Model Uncertainty and their Convergence}

\author[1]{J\"{o}rn Sass\thanks{\href{mailto:sass@mathematik.uni-kl.de}{sass@mathematik.uni-kl.de}}}
\author[1]{Dorothee Westphal\thanks{\href{mailto:westphal@mathematik.uni-kl.de}{westphal@mathematik.uni-kl.de}}}
\affil[1]{Department of Mathematics, Technische Universit\"{a}t Kaiserslautern}

\date{November~2, 2021}

\maketitle

\begin{abstract}
	In this paper we investigate a utility maximization problem with drift uncertainty in a multivariate continuous-time Black--Scholes type financial market which may be incomplete. We impose a constraint on the admissible strategies that prevents a pure bond investment and we include uncertainty by means of ellipsoidal uncertainty sets for the drift.
	Our main results consist firstly in finding an explicit representation of the optimal strategy and the worst-case parameter, secondly in proving a minimax theorem that connects our robust utility maximization problem with the corresponding dual problem.
	Thirdly, we show that, as the degree of model uncertainty increases, the optimal strategy converges to a generalized uniform diversification strategy.
	
	\medskip
	
	
	\smallskip
	
\end{abstract}

\section{Introduction}\label{sec:introduction}

Model uncertainty is a challenge that is inherent in many applications of mathematical models. Optimization procedures in general take place under a particular model. This model, however, might be misspecified due to statistical estimation errors, incomplete information, biases, and for various other reasons. In that sense, any specified model must be understood as an approximation of the unknown ``true'' model. Difficulties arise since a strategy which is optimal under the approximating model might perform rather badly for the true model specifications.
A natural way to deal with model uncertainty is to consider worst-case optimization.

Model uncertainty, also called \emph{Knightian uncertainty} in reference to the seminal book by Knight~\cite{knight_1921}, has been addressed in numerous papers. Gilboa and Schmeidler~\cite{gilboa_schmeidler_1989} and Schmeidler~\cite{schmeidler_1989} formulate rigorous axioms on preference relations that account for risk aversion and uncertainty aversion. A robust utility functional in their sense is a mapping
\[ X\mapsto \inf_{Q\in\mathcal{Q}}\E_Q\bigl[U(X)\bigr], \]
where $U$ is a utility function and $\mathcal{Q}$ a convex set of probability measures.
Chen and Epstein~\cite{chen_epstein_2002} give a continuous-time extension of this multiple-priors utility. In Maccheroni et al.~\cite{maccheroni_marinacci_rustichini_2006} the authors thoroughly axiomatize the robust approach to utility maximization via so-called ambiguity-averse preferences.

Optimal investment decisions under such preferences are investigated in Quenez~\cite{quenez_2004} and Schied~\cite{schied_2005}. An extension of those results by means of a duality approach is given in Schied~\cite{schied_2007}.
Uncertainty about both drift and volatility in a continuous-time Brownian framework under multiple priors is studied by Lin and Riedel~\cite{lin_riedel_2014}. Further papers addressing drift uncertainty in financial markets are Garlappi et al.~\cite{garlappi_uppal_wang_2007} and Biagini and P\i nar~\cite{biagini_pinar_2017}. The latter also focuses on ellipsoidal uncertainty sets, as we do in this work.
Neufeld and Nutz~\cite{neufeld_nutz_2018} incorporate jumps of the price process by considering a L\'{e}vy processes setup.

A relation between model uncertainty and portfolio diversification is investigated in a recent paper by Pham et al.~\cite{pham_wei_zhou_2018}.
Pflug et al.~\cite{pflug_pichler_wozabal_2012} study a one-period risk minimization problem under model uncertainty and show convergence of the optimal strategy to the uniform diversification strategy. Our results generalize these findings to a continuous-time utility maximization problem and provide an explanation for the good performance of the uniform diversification strategy also in a continuous-time setting.

\bigskip

The optimization problem that we address here is a utility maximization problem in a continuous-time financial market. The most basic utility maximization problem in a Black--Scholes market is the Merton problem of maximizing expected utility of terminal wealth. It can be written in the form
\[ V(x_0) = \sup_{\pi\in\calA(x_0)} \E\bigl[U(X^\pi_T)\bigr], \]
where $U\colon\R_+\to\R$ is a utility function, $X^\pi_T$ denotes the terminal wealth achieved when using strategy $\pi$, and $\calA(x_0)$ is the class of admissible strategies starting with initial capital $x_0$.
Merton~\cite{merton_1969} solves this problem for power and logarithmic utility in a multivariate financial market model and gives a corresponding optimal strategy.
However, the setup of the problem assumes that an investor knows the market parameters, in particular the drift $\mu$ of asset returns. This is a rather unrealistic assumption since drift parameters are notoriously difficult to estimate.
To obtain strategies that are robust with respect to a possible misspecification of the drift we consider the worst-case optimization problem
\[ \overline{V}(x_0) = \adjustlimits \sup_{\pi\in\calA(x_0)} \inf_{\mu\in K} \E_\mu\bigl[U(X^\pi_T)\bigr]. \]
Here, we write $\E_\mu[\cdot]$ for the expectation with respect to a measure $\mathbb{P}^\mu$ under which the drift of the asset returns is $\mu\in\R^d$, with $d$ denoting the number of risky assets in the market. The set $K\subseteq\R^d$ is called the \emph{uncertainty set}.
Our aim is to study the structure of optimal strategies, as well as their asymptotic behavior as the uncertainty set $K$ increases. Since for large uncertainty, investors usually do not invest in the risky assets at all, we restrict the class of admissible strategies by imposing a constraint that prevents a pure bond investment. We focus on ellipsoidal uncertainty sets $K$, see~\eqref{eq:uncertainty_ellipsoid}.

Our main results consist firstly in finding an explicit representation of the optimal strategy and the worst-case drift parameter for the robust utility maximization problem with constrained strategies and ellipsoidal uncertainty sets. Secondly, by using this explicit representation, a minimax theorem of the form
\[ \adjustlimits \sup_{\pi\in\calA(x_0)} \inf_{\mu\in K} \E_\mu\bigl[U(X^\pi_T)\bigr] = \adjustlimits \inf_{\mu\in K} \sup_{\pi\in\calA(x_0)} \E_\mu\bigl[U(X^\pi_T)\bigr] \]
is proven.
Thirdly, we show that the optimal strategy converges to a generalized uniform diversification strategy. In case of $K$ being a ball, this is the equal weight strategy, corresponding to uniform diversification. This result is somewhat surprising since in the limit the optimal strategy does not depend on the volatility structure of the assets anymore. In that sense, our results help to explain the popularity of uniform diversification strategies by the presence of uncertainty in the model.

\bigskip

The paper is organized as follows. In Section~\ref{cha:a_robust_utility_maximization_problem} we state our multivariate, possibly incomplete, Black--Scholes type financial market model and introduce the robust utility maximization problem.
Our main results are given in Section~\ref{cha:a_duality_approach}, where we solve our optimization problem for power and logarithmic utility. The main idea is to solve the dual problem explicitly and to show then that the solution forms a saddle point of the problem. We give representations of the optimal strategy and the worst-case drift parameter and prove a minimax theorem.
In Section~\ref{cha:asymptotic_behavior_as_uncertainty_increases} we study the asymptotic behavior of the optimal strategy and the worst-case parameter as the degree of uncertainty goes to infinity. We show that the optimal strategy converges to a generalized uniform diversification strategy, where by uniform diversification we mean the equal weight or $1/d$ strategy for the investment in the risky assets. Furthermore, we analyze the influence of the investor's risk aversion on the speed of convergence and investigate measures for the performance of the optimal robust strategies.
Section~\ref{sec:outlook} gives an outlook on more general financial market models with stochastic drift processes for which we state a suitable problem formulation. Our results can then be used to derive an explicit representation of the optimal strategy as well as a minimax theorem also in the more general model.
For better readability, all proofs are collected in Appendix~\ref{app:proofs}.

\paragraph{Notation.}
We use the notation $I_d$ for the identity matrix in $\R^{d\times d}$ as well as $e_i$, $i=1,\dots,d$, for the $i$-th standard unit vector in $\R^d$, and $\ones$ for the vector in $\R^d$ containing a one in every component. We shortly write $\R_+=(0,\infty)$. By $\langle\cdot,\cdot\rangle$ we denote the scalar product on $\R^d\times\R^d$ with $\langle x,y\rangle=x^\transp y$ for $x,y\in\R^d$.
If $x\in\R^d$ is a vector, $\lVert x\rVert$ denotes the Euclidean norm of $x$.

\section{Robust Utility Maximization Problem}\label{cha:a_robust_utility_maximization_problem}

\subsection{Financial market model}

We consider a continuous-time financial market with one risk-free and various risky assets. By $T>0$ we denote some finite investment horizon. Let $(\Omega, \calF, \mathbb{F}, \mathbb{P})$ be a filtered probability space where the filtration $\mathbb{F}=(\calF_t)_{t\in[0,T]}$ satisfies the usual conditions. All processes are assumed to be $\mathbb{F}$-adapted.
The risk-free asset $S^0$ is of the form $S^0_t=\rme^{rt}$, $t\in[0,T]$, where $r\in\R$ is the constant risk-free interest rate.
Aside from the risk-free asset, investors can also invest in $d\geq 2$ risky assets. Their return process $R=(R^1,\dots,R^d)^\transp$ is defined by
\[ \rmd R_t = \nu\,\rmd t + \sigma\,\rmd W_t, \quad R_0=0, \]
where $W=(W_t)_{t\in[0,T]}$ is an $m$-dimensional Brownian motion under $\mathbb{P}$ with $m\geq d$, allowing for incomplete markets. Further, $\nu\in\R^d$ and $\sigma\in\R^{d\times m}$, where we assume that $\sigma$ has full rank equal to $d$.

We introduce model uncertainty by assuming that the true drift of the stocks is only known to be an element of some set $K\subseteq\R^d$ with $\nu\in K$ and that investors want to maximize their worst-case expected utility when the drift takes values within $K$. The value $\nu$ can be thought of as an estimate for the drift that was for instance obtained from historical stock prices. Changing the drift from $\nu$ to some $\mu\in K$ can be expressed by a change of measure. For this purpose, define the process $(Z^\mu_t)_{t\in[0,T]}$ by
\[ Z^\mu_t = \exp\Bigl(\theta(\mu)^\transp W_t -\frac{1}{2}\lVert\theta(\mu)\rVert^2 t\Bigr), \]
where $\theta(\mu)=\sigma^\transp(\sigma\sigma^\transp)^{-1}(\mu-\nu)$. We can then define a new measure $\mathbb{P}^\mu$ by setting $\frac{\rmd \mathbb{P}^\mu}{\rmd \mathbb{P}} = Z^\mu_T$. Note that since $\theta(\mu)$ is a constant, the process $(Z^\mu_t)_{t\in[0,T]}$ is a strictly positive martingale. Therefore, $\mathbb{P}^\mu$ is a probability measure that is equivalent to $\mathbb{P}$ and we obtain from Girsanov's Theorem that the process $(W^\mu_t)_{t\in[0,T]}$, defined by $W^\mu_t = W_t-\theta(\mu)t$, is a Brownian motion under $\mathbb{P}^\mu$. We can thus rewrite the return dynamics as
\begin{equation*}
	\rmd R_t = \nu\,\rmd t + \sigma\,\rmd W_t = \nu\,\rmd t + \sigma\bigl(\rmd W^\mu_t+\theta(\mu)\,\rmd t\bigr) = \mu\,\rmd t + \sigma\,\rmd W^\mu_t,
\end{equation*}
and see that a change of measure from $\mathbb{P}$ to $\mathbb{P}^\mu$ corresponds to changing the drift in the return dynamics from $\nu$ to $\mu$. We thus shortly write $\E_\mu[\cdot]$ for the expectation under measure $\mathbb{P}^\mu$ and $\E[\cdot]=\E_\nu[\cdot]$ for the expectation under our reference measure $\mathbb{P}=\mathbb{P}^\nu$.

An investor's trading decisions are described by a self-financing trading strategy $(\pi_t)_{t\in[0,T]}$ with values in $\R^d$. The entry $\pi^i_t$, $i=1, \dots, d$, is the proportion of wealth invested in asset $i$ at time $t$. The corresponding wealth process $(X^\pi_t)_{t\in[0,T]}$ given initial wealth $x_0>0$ can then be described by the stochastic differential equation
\[ \rmd X^\pi_t = X^\pi_t\Bigl( r\,\rmd t + \pi_t^\transp(\mu-r\ones)\,\rmd t + \pi_t^\transp \sigma\,\rmd W^\mu_t \Bigr), \quad X^\pi_0=x_0, \]
for any $\mu\in K$.
We require trading strategies to be $\mathbb{F}^R$-adapted, where $\mathbb{F}^R=(\calF^R_t)_{t\in[0,T]}$ for $\calF^R_t=\sigma((R_s)_{s\in[0,t]})$. The admissibility set is defined as
\[ \calA(x_0) = \biggl\{\pi=(\pi_t)_{t\in[0,T]} \;\bigg|\; \pi \text{ is } \mathbb{F}^R\text{-adapted}, \, X^\pi_0=x_0, \, \E_\mu\biggl[\int_0^T\! \lVert\sigma^\transp\!\pi_t\rVert^2\,\rmd t\biggr]<\infty \text{ for all } \mu\in K\biggr\}. \]
Our robust portfolio optimization problem can then be formulated as
\begin{equation}\label{eq:value_function_robust}
	\overline{V}(x_0) = \adjustlimits \sup_{\pi\in\calA(x_0)} \inf_{\mu\in K} \E_\mu\bigl[U_\gamma(X^\pi_T)\bigr],
\end{equation}
where $U_\gamma$ is a power or logarithmic utility function, i.e.\ $U_\gamma\colon\R_+\to\R$ for $\gamma\in(-\infty,1)$, where $U_\gamma(x)=\frac{x^\gamma}{\gamma}$ for $\gamma\neq 0$ denotes power utility and $U_0(x)=\log(x)$ logarithmic utility.

\subsection{Constraint on the admissible strategies}\label{sec:constraint_on_the_admissible_strategies}

In the following, our aim is to investigate problem~\eqref{eq:value_function_robust} in detail.
First, we make the observation that for a large degree of model uncertainty the trivial strategy $\pi\equiv 0$ becomes optimal both for logarithmic and for power utility. This result has been shown in a similar setting by Biagini and P\i nar~\cite[Sec.~3.1--3.2]{biagini_pinar_2017} who address in addition to the finite horizon setting also the case with an infinite time horizon.

\begin{proposition}\label{prop:invest_only_in_bond}
	Let $\gamma\in(-\infty,1)$ and $K\subseteq \R^d$. If $r\ones\in K$, then the strategy $(\pi_t)_{t\in[0,T]}$ with $\pi_t=0$ for all $t\in[0,T]$ is optimal for the optimization problem
	\begin{equation}\label{eq:recall_value_function_robust}
		\adjustlimits \sup_{\pi\in\calA(x_0)} \inf_{\mu\in K} \E_\mu\bigl[U_\gamma(X^\pi_T)\bigr].
	\end{equation}
\end{proposition}

This observation implies that as the level of uncertainty about the true drift parameter exceeds a certain threshold, it is optimal for investors to not invest anything in the stocks.

\begin{remark}
	Proposition~\ref{prop:invest_only_in_bond} could be reformulated in terms of robust utility functionals by assuming only that a martingale measure is in the ambiguity set.
	The statement of the proposition is in line with \O{}ksendal and Sulem~\cite{oksendal_sulem_2008, oksendal_sulem_2011} where the authors obtain a similar result for optimality of $\pi\equiv 0$. They consider a jump diffusion model with a worst-case approach where the market chooses a scenario from a fixed but very comprehensive set of probability measures. In contrast, it is shown in Zawisza~\cite{zawisza_2018} that, if the model allows for stochastic interest rate $r$, the optimal strategy does not invest exclusively in the bond.
	Lin and Riedel~\cite{lin_riedel_2021} show that, when there is a large degree of uncertainty about interest rates, the investor even puts all money in the risky assets.
\end{remark}

Investing everything in the risk-free asset is a sensible but very extreme reaction to model uncertainty. We are interested in finding out which strategies are reasonable under high model uncertainty if investors still want to invest a part of their wealth into the risky assets, or, alternatively, if they are forced to invest due to some external requirements. For that purpose, we introduce a constraint on our strategies that prevents investors from solely investing in the bond.
Consider for some $h>0$ the admissibility set
\[ \calA_h(x_0)=\bigl\{ \pi\in\calA(x_0) \,\big|\, \langle\pi_t,\ones\rangle = h \text{ for all } t\in[0,T] \bigr\}. \]
We do not want to exclude short-selling, so negative entries of $\pi$ are possible.
Taking $h=1$ would imply that investors are not allowed to invest anything in the risk-free asset. They must then distribute all of their wealth among the risky assets. For instance, a constraint of the form $\langle\pi_t,\ones\rangle = h>0$ typically applies for some mutual funds when investors are required to invest a certain amount in risky assets.
Moreover, it has been studied in DeMiguel et al.~\cite{demiguel_garlappi_nogales_uppal_2009} how constraining the norm of portfolio weight vectors in a one-period model can improve portfolio performance in the presence of estimation errors.

\begin{remark}
	The admissibility set $\calA_h(x_0)$ might seem unnecessarily restrictive at first glance. Instead of fixing $\langle\pi_t,\ones\rangle=h$ one might want to consider utility maximization among the larger class of strategies $\pi$ with $\langle\pi_t,\ones\rangle\geq h$. However, we are mainly interested in the asymptotic behavior of the optimal strategies as the level of uncertainty increases. It is intuitively clear that, when uncertainty is large, investors seek to invest as little as possible in the risky assets.
	Therefore, we consider optimization among strategies in $\calA_h(x_0)$ and use our results to show that enlarging the class of admissible strategies asymptotically does not change the value of the optimization problem, see Section~\ref{sec:relaxing_the_investment_constraint}.
\end{remark}

\section{A Duality Approach}\label{cha:a_duality_approach}

In this section we solve for power or logarithmic utility $U_\gamma$ and for specific uncertainty sets $K$ the optimization problem
\begin{equation}\label{eq:robust_problem_power_log}
	\adjustlimits \sup_{\pi\in\calA_h(x_0)} \inf_{\mu\in K} \E_\mu\bigl[U_\gamma(X^\pi_T)\bigr].
\end{equation}

\begin{remark}
	In the situation with logarithmic utility and uncertainty sets that are balls in some $p$-norm, $p\in[1,\infty)$, it is possible to carry over methods from a one-period risk minimization problem as in Pflug et al.~\cite{pflug_pichler_wozabal_2012} to our continuous-time robust utility maximization problem.
	If $K=\{\mu\in\R^d\,|\,\lVert\mu-\nu\rVert_p\leq\kappa\}$, then for every $\varepsilon>0$ there exists a $\kappa_0>0$ such that for all $\kappa\geq\kappa_0$ the strategy $\pi^*(\kappa)$ that is optimal for
	\[ \adjustlimits \sup_{\substack{\pi\in\calA_h(x_0)\\ \pi \text{ deterministic}}} \inf_{\mu\in K} \E_\mu\bigl[\log(X^\pi_T)\bigr] \]
	satisfies
	\[ \biggl\lVert\frac{1}{T}\int_0^T\Bigl(\pi_s^*(\kappa)-\frac{h}{d}\ones\Bigr)\rmd s\biggr\rVert_q<\varepsilon, \]
	where $q\in(1,\infty]$ with $\frac{1}{p}+\frac{1}{q}=1$. See Westphal~\cite[Thm.~3.4]{westphal_2019} for a proof. This shows that the optimal strategy among the deterministic ones converges, as model uncertainty increases, to a uniform diversification strategy $\pi^u$ with $\pi^u_t=\frac{h}{d}\ones$ for every $t\in[0,T]$. Hence, as uncertainty about the true drift parameter goes to infinity, investors split the proportion $h$ of their money more and more evenly among all risky assets.
	
	This approach has several drawbacks. Firstly, we can follow the ideas from Pflug et al.~\cite{pflug_pichler_wozabal_2012} in continuous time only for logarithmic utility and uncertainty sets $K$ that are balls in $p$-norm.
	Secondly, we have to restrict to the class of deterministic strategies to be able to use their methods. However, it is by no means clear in the first place that an optimal strategy to our problem should be a deterministic one. In fact, in many worst-case optimization problems it is even beneficial to use randomized strategies, see Delage et al.~\cite{delage_kuhn_wiesemann_2019}.
	And lastly, the above result does not yield an explicit solution to the robust optimization problem, it only gives asymptotic results for large levels of uncertainty.
	To overcome these problems we follow here a different approach that works for both power and logarithmic utility and that results in an explicit solution of the optimization problem.
\end{remark}

We study the case where the uncertainty set is an ellipsoid in $\R^d$ centered around the reference parameter $\nu$, i.e.\
\begin{equation}\label{eq:uncertainty_ellipsoid}
	K=\bigl\{ \mu\in\R^d \,\big|\, (\mu-\nu)^\transp \Gamma^{-1}(\mu-\nu) \leq \kappa^2 \bigr\}.
\end{equation}
Here, $\kappa>0$, $\nu\in\R^d$, and $\Gamma\in\R^{d\times d}$ is symmetric and positive definite. The matrix $\Gamma$ determines the shape of the ellipsoid, the value of $\kappa$ its size. Higher values of $\kappa$ correspond to more uncertainty about the true drift.

By means of $\Gamma$ we can model that some (linear combinations of) drifts are known at a higher degree of accuracy than others. A special case discussed in the literature is $\Gamma=\sigma\sigma^\transp$, see e.g.\ Biagini and P\i nar~\cite{biagini_pinar_2017}. But also different forms of $\Gamma$ can be motivated. For $\Gamma=I_d$ we simply get a ball in the Euclidean norm with radius $\kappa$ and center $\nu$. By setting $\Gamma$ equal to a diagonal matrix different from the identity we can give different weights to the uncertainty of the single asset drifts.

More generally, assume that the reference drift parameter $\nu$ is obtained as the value of an unbiased estimator $\hat{\mu}$ for the true drift, say from observing historical returns. Then the covariance matrix $\cov(\hat{\mu})$ is a reasonable choice for $\Gamma$, because then the uncertainty set $K$ constitutes a natural (asymptotic) confidence region for the true drift. This flexibility in the form of $\Gamma$ is especially useful for a generalization of our model to a setting with time-dependent drift and uncertainty sets, see Section~\ref{sec:outlook}, where we give a short outlook on Sass and Westphal~\cite{sass_westphal_2021}. In that follow-up work a time-dependent uncertainty set is constructed based on filtering techniques.

\subsection{Solution of the non-robust problem}\label{subs:solution_of_the_non-robust_problem}

To solve the optimization problem~\eqref{eq:robust_problem_power_log} we first address the non-robust constrained utility maximization problem under a fixed parameter $\mu\in\R^d$.
We repeatedly make use of a specific matrix that we introduce in the following lemma.

\begin{lemma}\label{lem:definition_D}
	Consider the matrix
	\[ D =
	\begin{pmatrix}
		1	&	& 0	& -1 \\
			&\ddots	&	&\vdots \\
		0	&	& 1	& -1
	\end{pmatrix}\in\R^{(d-1)\times d}. \]
	Then, given that $\sigma\in\R^{d\times m}$ has rank $d$, $D\sigma$ has rank $d-1$.
\end{lemma}

The matrix $D$ defined in the lemma above comes up naturally in calculations when using the constraint $\langle \pi_t,\ones\rangle=h$ in the form $\pi^d_t = h-\sum_{i=1}^{d-1} \pi^i_t$. This can be seen as a reduction of the problem from $d$ dimensions to $d-1$ dimensions.
For better readability of the calculations below we introduce the following notation.

\begin{definition}\label{def:matrix_A_vector_c}
	We define the matrix $A\in\R^{d\times d}$ and the vector $c\in\R^d$ by
	\begin{align*}
		A &= D^\transp(D\sigma\sigma^\transp D^\transp)^{-1}D, \\
		c &= e_d-D^\transp(D\sigma\sigma^\transp D^\transp)^{-1}D\sigma\sigma^\transp e_d = (I_d-A\sigma\sigma^\transp)e_d,
	\end{align*}
	where $D\in\R^{(d-1)\times d}$ is as given in Lemma~\ref{lem:definition_D} and $e_d$ is the $d$-th standard unit vector in $\R^d$.
\end{definition}

Note that we assume $\sigma\in\R^{d\times m}$ to have full rank, hence by the previous lemma we know that $D\sigma$ has full rank, in particular $D\sigma\sigma^\transp D^\transp=D\sigma(D\sigma)^\transp$ is nonsingular.
Using this notation we give the optimal strategy for the constrained optimization problem given a fixed drift $\mu$.
The possible incompleteness of the market does not complicate our approach here. The reason is that, for determining the optimal strategy, we can essentially reduce the problem to an unconstrained less-dimensional financial market where the optimal strategy can be obtained as a classical Merton strategy.

\begin{proposition}\label{prop:optimal_strategy_non-robust}
	Let $\mu\in\R^d$. Then the optimal strategy for the optimization problem
	\[ \sup_{\pi\in\calA_h(x_0)} \E_\mu\bigl[U_\gamma(X^\pi_T)\bigr] \]
	is the strategy $(\pi_t)_{t\in[0,T]}$ with
	\[ \pi_t = \frac{1}{1-\gamma}A\mu +hc \]
	for all $t\in[0,T]$, with $A$ and $c$ as in Definition~\ref{def:matrix_A_vector_c}.
\end{proposition}
      
In the proof the $d$-dimensional constrained problem is reduced to a $(d-1)$-dimensional unconstrained problem. Using the form of the optimal strategy in the $(d-1)$-dimensional market which is known from Merton~\cite{merton_1969} yields the following representation for the optimal expected utility from terminal wealth.

\begin{corollary}\label{cor:optimal_utility_non-robust}
	Let $\mu\in\R^d$. Then the optimal expected utility from terminal wealth is
	\begin{equation*}
		\begin{aligned}
			\sup_{\pi\in\calA_h(x_0)} &\E_\mu\bigl[U_\gamma(X^\pi_T)\bigr] \\
			&=
			\begin{dcases}
				\frac{x_0^\gamma}{\gamma}\exp\Bigl(\gamma T\Bigl( \widetilde{r}+\frac{1}{2(1-\gamma)}\bigl(\widetilde{\mu}-\widetilde{r}\mathbf{1}_{d-1}\bigr)^\transp(\widetilde{\sigma}\widetilde{\sigma}^\transp )^{-1}\bigl(\widetilde{\mu}-\widetilde{r}\mathbf{1}_{d-1}\bigr)\Bigr)\Bigr), &\gamma\neq 0,\\
				\log(x_0) + \Bigl( \widetilde{r}+\frac{1}{2}\bigl(\widetilde{\mu}-\widetilde{r}\mathbf{1}_{d-1}\bigr)^\transp(\widetilde{\sigma}\widetilde{\sigma}^\transp )^{-1}\bigl(\widetilde{\mu}-\widetilde{r}\mathbf{1}_{d-1}\bigr) \Bigr)T, &\gamma=0,
			\end{dcases}
		\end{aligned}
	\end{equation*}
	where
	\begin{equation}\label{eq:recall_substitution_r_mu_sigma}
		\begin{aligned}
			\widetilde{\sigma}&=D\sigma, \\
			\widetilde{r}&=(1-h)r+he_d^\transp\mu-\frac{1}{2}(1-\gamma)\lVert h\sigma^\transp e_d \rVert^2, \\
			\widetilde{\mu}&=D\mu - h(1-\gamma)D\sigma\sigma^\transp e_d+\widetilde{r}\mathbf{1}_{d-1}.
		\end{aligned}
	\end{equation}
\end{corollary}

The previous results give a representation of the optimal strategy and the optimal expected utility of terminal wealth under the constraint $\langle\pi_t,\ones\rangle = h$, given that the drift parameter $\mu$ is known. Of course, both the strategy and the terminal wealth then depend on $\mu$. However, we aim at solving the robust utility maximization problem
\[ \adjustlimits \sup_{\pi\in\calA_h(x_0)} \inf_{\mu\in K} \E_\mu\bigl[U_\gamma(X^\pi_T)\bigr]. \]
For that purpose, we address in a next step the question what the worst possible parameter $\mu$ would be for the investor, given that she reacts optimally, i.e.\ by applying the strategy from Proposition~\ref{prop:optimal_strategy_non-robust}. This corresponds to solving the dual problem
\[ \adjustlimits \inf_{\mu\in K} \sup_{\pi\in\calA_h(x_0)} \E_\mu\bigl[U_\gamma(X^\pi_T)\bigr]. \]
Note here that we do not know yet whether equality holds between our original problem and the corresponding dual problem. In general the solution of the dual problem may not be of great help. In the following, after deriving the solution to the dual problem, we prove a minimax theorem that establishes the desired equality. Results from the literature, e.g.\ from Quenez~\cite{quenez_2004}, do not directly carry over to our setting as we discuss in Remark~\ref{rem:minimax_theorems_in_literature} below.

\subsection{The worst-case parameter}\label{subs:the_worst-case_parameter}

From Corollary~\ref{cor:optimal_utility_non-robust} we have a representation of the optimal expected utility of terminal wealth, depending on the transformed parameters $\widetilde{r}$, $\widetilde{\mu}$ and $\widetilde{\sigma}$. Note that for any $\gamma\in(-\infty,1)$, minimizing this expression in $\mu$ is equivalent to minimizing
\[ \widetilde{r}+\frac{1}{2(1-\gamma)}\bigl(\widetilde{\mu}-\widetilde{r}\mathbf{1}_{d-1}\bigr)^\transp(\widetilde{\sigma}\widetilde{\sigma}^\transp )^{-1}\bigl(\widetilde{\mu}-\widetilde{r}\mathbf{1}_{d-1}\bigr). \]
We now plug in the representations of $\widetilde{r}$, $\widetilde{\mu}$ and $\widetilde{\sigma}$ from the corollary and obtain
\begin{multline}
	(1-h)r+he_d^\transp\mu-\frac{1}{2}(1-\gamma)\lVert h\sigma^\transp e_d \rVert^2\\
	+\frac{1}{2(1-\gamma)}\bigl(D\mu - h(1-\gamma)D\sigma\sigma^\transp e_d\bigr)^\transp(D\sigma\sigma^\transp D^\transp)^{-1}\bigl(D\mu - h(1-\gamma)D\sigma\sigma^\transp e_d\bigr).
\end{multline}
Our aim is to minimize the above expression in $\mu$. We see that many terms do not depend on $\mu$. The minimization is therefore equivalent to the minimization of
\begin{equation}\label{eq:what_is_to_be_minimized_in_mu}
	\begin{aligned}
		&he_d^\transp\mu+\frac{1}{2(1-\gamma)}\Bigl( \mu^\transp\! D^\transp(D\sigma\sigma^\transp\! D^\transp)^{-1}D\mu -2h(1-\gamma)(D\sigma\sigma^\transp\! e_d)^\transp(D\sigma\sigma^\transp\! D^\transp)^{-1}D\mu\Bigr)\\
		&=\frac{1}{2(1-\gamma)}\mu^\transp D^\transp(D\sigma\sigma^\transp D^\transp)^{-1}D\mu + h\Bigl( e_d^\transp\mu-(D\sigma\sigma^\transp e_d)^\transp(D\sigma\sigma^\transp D^\transp)^{-1}D\mu \Bigr) \\
		&= \frac{1}{2(1-\gamma)}\mu^\transp A\mu+hc^\transp\mu
	\end{aligned}
\end{equation}
on the ellipsoid $K$, where $A$ and $c$ were introduced in Definition~\ref{def:matrix_A_vector_c}.
To make this minimization problem easier, we apply a transformation to the elements $\mu\in K$. For that purpose, note that since $\Gamma\in\R^{d\times d}$ is assumed to be symmetric and positive definite, there exists some nonsingular matrix $\tau\in\R^{d\times d}$ such that $\Gamma=\tau\tau^\transp$. The matrix $\tau$ can be obtained for example by the Cholesky decomposition. Then we can rewrite the constraint $(\mu-\nu)^\transp \Gamma^{-1}(\mu-\nu) \leq \kappa^2$ as
\[ \kappa^2\geq (\mu-\nu)^\transp(\tau\tau^\transp)^{-1}(\mu-\nu)=(\mu-\nu)^\transp(\tau^\transp)^{-1}\tau^{-1}(\mu-\nu)=\bigl(\tau^{-1}(\mu-\nu)\bigr)^\transp\bigl(\tau^{-1}(\mu-\nu)\bigr). \]
Hence, for an arbitrary $\mu\in K$ we define $\rho:=\tau^{-1}(\mu-\nu)$ so that $\mu=\nu+\tau\rho$ and $\lVert\rho\rVert\leq\kappa$. We can then rewrite~\eqref{eq:what_is_to_be_minimized_in_mu} as
\begin{equation*}
	\begin{aligned}
		\frac{1}{2(1-\gamma)}\mu^\transp A\mu&+hc^\transp\mu
		= \frac{1}{2(1-\gamma)}\bigl((\tau\rho)^\transp A\tau\rho+2\nu^\transp A\tau\rho+\nu^\transp A\nu\bigr) +hc^\transp\tau\rho +hc^\transp\nu \\
		&= \frac{1}{2(1-\gamma)}\rho^\transp\tau^\transp A\tau\rho+\Bigl(\frac{1}{1-\gamma}A\nu+hc\Bigr)^\transp\tau\rho+\frac{1}{2(1-\gamma)}\nu^\transp A\nu +hc^\transp\nu.
	\end{aligned}
\end{equation*}
Minimizing~\eqref{eq:what_is_to_be_minimized_in_mu} in $\mu\in K$ is therefore equivalent to minimizing the function $g\colon B_\kappa(0)\to\R$ with
\[ g(\rho)=\frac{1}{2(1-\gamma)}\rho^\transp \tau^\transp A\tau\rho+\Bigl(hc+\frac{1}{1-\gamma}A\nu\Bigr)^\transp\tau\rho \]
in $\rho$ and then setting $\mu=\nu+\tau\rho$.
The behavior of $g$ is determined to a large extent by the matrix $A$ from Definition~\ref{def:matrix_A_vector_c}. So we analyze properties of $A$ next.

\begin{lemma}\label{lem:A_symmetric_positive_semidefinite}
	The matrix $A$ is symmetric and positive semidefinite with $\mathrm{ker}(A)=\mathrm{span}(\{\ones\})$.
\end{lemma}

We immediately deduce that also $\tau^\transp A\tau\in\R^{d\times d}$ is symmetric and positive semidefinite with $\mathrm{ker}(\tau^\transp A\tau)=\mathrm{span}(\{\tau^{-1}\ones\})$. Having collected these properties of the matrix $A$ and of $\tau^\transp A\tau$ enables us to find the parameter $\rho$ that minimizes $g(\rho)$ on the set $B_\kappa(0)$.

\begin{lemma}\label{lem:worst-case_parameter}
	Let $0=\lambda_1<\lambda_2\leq\cdots\leq\lambda_d$ denote the eigenvalues of $\tau^\transp A\tau$, and let further $v_1=\frac{1}{\lVert \tau^{-1}\ones \rVert}\tau^{-1}\ones, v_2,\dots,v_d\in\R^d$ denote the respective orthogonal eigenvectors with $\lVert v_i\rVert=1$ for all $i=1,\dots, d$.
	Then the minimum of the function $g\colon B_\kappa(0)\to\R$ with
	\[ g(\rho)=\frac{1}{2(1-\gamma)}\rho^\transp \tau^\transp A\tau\rho+\Bigl(hc+\frac{1}{1-\gamma}A\nu\Bigr)^\transp\tau\rho \]
	on the domain $B_\kappa(0)=\{ \rho\in\R^d \,|\, \lVert \rho\rVert \leq \kappa \}$ is attained by the vector
	\[ \rho^*=-\sum_{i=1}^d \biggl({\frac{\lambda_i}{1-\gamma}+\frac{h}{\psi(\kappa)\lVert \tau^{-1}\ones \rVert}}\biggr)^{-1}\biggl\langle h\tau^\transp c+\frac{\lambda_i}{1-\gamma}\tau^{-1}\nu, v_i\biggr\rangle v_i, \]
	where $\psi(\kappa)\in(0,\kappa]$ is uniquely determined by $\lVert\rho^*\rVert=\kappa$.
\end{lemma}

Note that $\psi(\kappa)$ in the above lemma is the unique value in $(0,\kappa]$ that makes $\rho^*$ lie on the boundary of $B_\kappa(0)$. In the representation $\rho^*=\sum_{i=1}^d a_iv_i$ it holds $a_1=-\psi(\kappa)$, i.e.\ $\psi(\kappa)$ is the negative of the coefficient belonging to $v_1$. Recall that $v_1$ is the eigenvector to eigenvalue zero of $\tau^\transp A\tau$, hence it plays an important role in the minimization of the function $g$ above.
In Section~\ref{cha:asymptotic_behavior_as_uncertainty_increases} we will study the asymptotic behavior for large uncertainty $\kappa$. It will turn out that asymptotically $v_1$ will be the dominant component in the representation $\rho^*=\sum_{i=1}^d a_iv_i$, a claim that we show by analyzing the asymptotic behavior of $\psi(\kappa)$.
The previous lemma now yields the solution of the dual problem to our original optimization problem.

\begin{theorem}\label{thm:solution_of_the_inf_sup_problem}
	Let $0=\lambda_1<\lambda_2\leq\cdots\leq\lambda_d$ denote the eigenvalues of $\tau^\transp A\tau$, and let further $v_1=\frac{1}{\lVert \tau^{-1}\ones \rVert}\tau^{-1}\ones, v_2,\dots,v_d\in\R^d$ denote the respective orthogonal eigenvectors with $\lVert v_i\rVert=1$ for all $i=1,\dots, d$.
	Then
	\[ \adjustlimits \inf_{\mu\in K} \sup_{\pi\in\calA_h(x_0)} \E_\mu\bigl[U_\gamma(X^\pi_T)\bigr] = \E_{\mu^*}\bigl[U_\gamma(X^{\pi^*}_T)\bigr], \]
	where
	\[ \mu^*=\nu-\tau\sum_{i=1}^d \biggl( \frac{\lambda_i}{1-\gamma}+\frac{h}{\psi(\kappa)\lVert \tau^{-1}\ones \rVert} \biggr)^{-1}\biggl\langle h\tau^\transp c+\frac{\lambda_i}{1-\gamma}\tau^{-1}\nu, v_i\biggr\rangle v_i \]
	for $\psi(\kappa)\in(0,\kappa]$ that is uniquely determined by $\lVert\tau^{-1}(\mu^*-\nu)\rVert=\kappa$, and where $(\pi^*_t)_{t\in[0,T]}$ is for all $t\in[0,T]$ defined by
	\[ \pi^*_t = \frac{1}{1-\gamma}A\mu^* +hc. \]
\end{theorem}

\begin{remark}\label{rem:minimax_theorems_in_literature}
	The preceding theorem solves the problem
	\begin{equation}\label{eq:the_inf_sup_problem}
		\adjustlimits \inf_{\mu\in K} \sup_{\pi\in\calA_h(x_0)} \E_\mu\bigl[U_\gamma(X^\pi_T)\bigr].
	\end{equation}
	This is the corresponding dual problem to our original optimization problem
	\begin{equation}\label{eq:the_sup_inf_problem}
		\adjustlimits \sup_{\pi\in\calA_h(x_0)} \inf_{\mu\in K} \E_\mu\bigl[U_\gamma(X^\pi_T)\bigr],
	\end{equation}
	but in general the values of these two problems do not coincide. There are, of course, special cases in which the supremum and the infimum do interchange. Those results are called \emph{minimax theorems} in the literature. In a portfolio optimization setting that is similar to ours a minimax theorem has been shown in Quenez~\cite{quenez_2004}. Here, the author applies classical techniques from Kramkov and Schachermayer~\cite{kramkov_schachermayer_1999, kramkov_schachermayer_2003} for incomplete markets and embeds them into a multiple-priors framework. However, there are two main points that distinguish our setting from the one in Quenez~\cite{quenez_2004}. Firstly, the results in that paper are only shown for non-negative utility functions and therefore not directly applicable to power utility $U_\gamma$ with a negative $\gamma$. Secondly, the constraint $\langle\pi_t,\ones\rangle=h$ that we put on the admissible trading strategies alters the structure of attainable terminal wealths so that it would be necessary to adjust the proofs and check several technical assumptions.
	
	In addition, note that a minimax theorem does not endow us with the form of the optimal strategy (or the worst-case drift) yet. To obtain an explicit representation of the same, we would still need to go through the calculations done in this section. In the following, we will use the explicit representation of the optimal strategy for~\eqref{eq:the_inf_sup_problem} to show that it indeed also solves~\eqref{eq:the_sup_inf_problem} and that in this case, the supremum and the infimum can be interchanged.
\end{remark}

\subsection{A minimax theorem}\label{subs:a_minimax_theorem}

The following representation of $\pi^*$ is useful for proving our minimax theorem.

\begin{lemma}\label{lem:representation_of_pi_star}
	The strategy $\pi^*$ from Theorem~\ref{thm:solution_of_the_inf_sup_problem} satisfies
	\[ \pi^*_t = -\frac{h}{\psi(\kappa)\lVert \tau^{-1}\ones \rVert}\Gamma^{-1}(\mu^*-\nu) \]
	for all $t\in[0,T]$.
\end{lemma}

The preceding lemma characterizes the strategy $\pi^*$, which is the best strategy an investor can choose when the drift of stocks is $\mu^*$. In the following we show that, vice versa, $\mu^*$ is also the parameter the market has to choose to minimize the investor's expected utility of terminal wealth, given that the investor applies strategy $\pi^*$. It then follows that the point $(\pi^*,\mu^*)$ is a \emph{saddle point} of our problem, i.e.\ it holds
\[ \E_{\mu^*}\bigl[U_\gamma(X^{\pi}_T)\bigr] \leq \E_{\mu^*}\bigl[U_\gamma(X^{\pi^*}_T)\bigr] \leq \E_{\mu}\bigl[U_\gamma(X^{\pi^*}_T)\bigr] \]
for all $\mu\in K$ and $\pi\in\calA_h(x_0)$. This property is essential for proving our minimax theorem. Note that the inequality
\[ \adjustlimits \sup_{\pi\in\calA_h(x_0)} \inf_{\mu\in K} \E_\mu\bigl[U_\gamma(X^\pi_T)\bigr] \leq \adjustlimits \inf_{\mu\in K} \sup_{\pi\in\calA_h(x_0)} \E_\mu\bigl[U_\gamma(X^\pi_T)\bigr] \]
always holds when interchanging supremum and infimum, see Ekeland and Temam~\cite[Ch.~VI, Prop.~1.1]{ekeland_temam_1976}, for example. For the reverse inequality the saddle point property is needed.

\begin{theorem}\label{thm:duality_result}
	Let $K=\{ \mu\in\R^d \,|\, (\mu-\nu)^\transp \Gamma^{-1}(\mu-\nu) \leq \kappa^2 \}$. Then the parameter $\mu$ that attains the minimum in
	\[ \inf_{\mu\in K} \E_\mu\bigl[U_\gamma(X^{\pi^*}_T)\bigr] \]
	is $\mu^*$, where both $\mu^*$ and $\pi^*$ are defined as in Theorem~\ref{thm:solution_of_the_inf_sup_problem}.
	In particular, it follows that
	\[ \adjustlimits \sup_{\pi\in\calA_h(x_0)} \inf_{\mu\in K} \E_\mu\bigl[U_\gamma(X^\pi_T)\bigr] = \E_{\mu^*}\bigl[U_\gamma(X^{\pi^*}_T)\bigr] = \adjustlimits \inf_{\mu\in K} \sup_{\pi\in\calA_h(x_0)} \E_\mu\bigl[U_\gamma(X^\pi_T)\bigr]. \]
\end{theorem}

The previous theorem establishes duality between our original robust utility maximization problem and the dual problem where supremum and infimum are interchanged. Additionally, we now also know the solution to our original problem. The optimal strategy for our constrained robust utility maximization problem is given in a nearly explicit way. Note that the parameter $\mu^*$ in Theorem~\ref{thm:solution_of_the_inf_sup_problem} is not given explicitly since the parameter $\psi(\kappa)$ is defined in an implicit way. However, finding $\psi(\kappa)$ numerically can be done in a straightforward way by a numerical root search of a monotone function. For this reason, determining $\mu^*$ and $\pi^*$ numerically does not pose any problems.

\begin{remark}
	One can think of other reasonable sets $K$ for modelling uncertainty about the drift parameter $\mu$. Our duality approach can also be applied to the optimization problem with
	\[ K=\bigl\{\mu\in\R^d \,\big|\, \ones^\transp\mu=b\bigr\} \]
	for some $b\in\R$. The motivation for this uncertainty set is that one has an estimate for the performance of a stock index, and therefore the overall average performance of the stocks, but not for the single stocks themselves. In that case, one can show that the optimal strategy for the optimization problem
	\[ \adjustlimits \inf_{\mu\in K} \sup_{\pi\in\calA_h(x_0)} \E_\mu\bigl[U_\gamma(X^\pi_T)\bigr] \]
	is $(\pi^*_t)_{t\in[0,T]}$ with $\pi^*_t = \frac{h}{d}\ones$ for all $t\in[0,T]$. The worst-case parameter $\mu^*$ can be determined explicitly given the eigenvalues and eigenvectors of the matrix $A$.
	Further, one can show a minimax theorem in analogy to Theorem~\ref{thm:duality_result}.
	The optimal strategy is here just a uniform diversification strategy given the constraint on the bond investment. In the next section we show how this fits into the framework of our results for ellipsoidal uncertainty sets when we let the degree of uncertainty $\kappa$ go to infinity.
\end{remark}

\section{Asymptotic Behavior as Uncertainty Increases}\label{cha:asymptotic_behavior_as_uncertainty_increases}

In this section we consider again the setting with ellipsoidal uncertainty sets as in~\eqref{eq:uncertainty_ellipsoid} and investigate what happens as the degree of uncertainty changes. Since $K$ is an ellipsoid, we increase the degree of uncertainty about the true drift parameter by increasing the radius $\kappa$, a lower value of $\kappa$ corresponds to a more precise knowledge of the true drift.

\subsection{Limit of worst-case parameter and optimal strategy}\label{sec:limit_of_worst-case_parameter_and_optimal_strategy}

In the following, we address in detail the asymptotic behavior of the worst-case parameter and the optimal strategy as uncertainty increases, i.e.\ as $\kappa$ goes to infinity. To underline the dependence on the degree of uncertainty, we write $\mu^*=\mu^*(\kappa)$ and $\pi^*=\pi^*(\kappa)$ in the following.

\begin{remark}
	The other asymptotic regime $\kappa\to 0$ corresponds to a more and more precise knowledge of the true drift. It is easy to see that
	\[ \lim_{\kappa\to 0}\mu^*(\kappa)=\nu \quad\text{and}\quad \lim_{\kappa\to 0}\pi_t^*(\kappa)=\frac{1}{1-\gamma}A\nu+hc \]
	for all $t\in[0,T]$. This means that the worst-case parameter converges to the reference drift $\nu$ and the optimal strategy to the best constrained strategy, given that the drift equals $\nu$. So we retrieve in the limit $\kappa\to 0$ the setting without model uncertainty.
\end{remark}

We now focus on $\kappa\to\infty$. Note that the only quantity in the representation of $\mu^*$ from Theorem~\ref{thm:solution_of_the_inf_sup_problem} that depends on $\kappa$ is $\psi(\kappa)$.

\begin{lemma}\label{lem:asymptotics_of_a_kappa}
	It holds $\lim_{\kappa\to\infty} \frac{\psi(\kappa)}{\kappa}=1$.
\end{lemma}

From this lemma we gain insights into the asymptotic behavior of $\mu^*$.

\begin{proposition}\label{prop:asymptotics_of_mu_star}
	It holds
	\[ \lim_{\kappa\to\infty} \frac{1}{\kappa}\tau^{-1}\bigl(\mu^*(\kappa)-\nu\bigr)=-v_1=-\frac{1}{\lVert \tau^{-1}\ones \rVert}\tau^{-1}\ones \]
	and
	\[ \lim_{\kappa\to\infty} \frac{1}{\kappa}\mu^*(\kappa) = -\tau v_1 = -\frac{1}{\lVert \tau^{-1}\ones \rVert}\ones. \]
\end{proposition}

Hence, asymptotically the direction of the worst-case parameter is $-\ones$. This means that, as $\kappa$ tends to infinity, the worst drift which the market can choose for an investor who applies the optimal strategy $\pi^*$, is a drift vector where all entries are the same and negative. We have the following result for the asymptotic behavior of the investor's optimal strategy.

\begin{theorem}\label{thm:limit_of_optimal_strategy}
	For any $t\in[0,T]$ it holds
	\[ \lim_{\kappa\to\infty} \pi^*_t(\kappa)=\frac{h}{\ones^\transp\Gamma^{-1}\ones}\Gamma^{-1}\ones. \]
\end{theorem}

The theorem shows that the optimal strategy $\pi^*(\kappa)$ converges as the degree of uncertainty $\kappa$ goes to infinity. If $\Gamma=\sigma\sigma^\transp$, then the limit is a multiple of the minimum variance portfolio. Another interesting special case is $\Gamma=I_d$, i.e.\ when $K$ is simply a ball with radius $\kappa$.
In that case we have
\[ \lim_{\kappa\to\infty} \pi^*_t(\kappa)=\frac{h}{d}\ones \]
for any $t\in[0,T]$, hence the optimal strategy converges to a uniform diversification strategy, given by $\frac{h}{d}\ones$ at each point in time. Hence, when forced to invest a total fraction of $h>0$ in the risky assets, then in the limit for $\kappa$ going to infinity investors will diversify their portfolio uniformly. For general $\Gamma$ we shall speak of a generalized uniform diversification strategy.

This asymptotic behavior of the optimal strategy is striking because the limit is independent of the volatility matrix $\sigma$. In combination with the structure of the function $g$ in Lemma~\ref{lem:worst-case_parameter} this indicates that it might also be possible to allow for misspecified volatility. For a high level of uncertainty the optimal strategy is dominated by the matrix $\Gamma$ shaping the uncertainty ellipsoid whereas both the volatility structure of the assets and the reference drift $\nu$ become negligible. This effect is caused by the investor's reaction to the worst-case drift parameter $\mu^*$ which, as shown in Proposition~\ref{prop:asymptotics_of_mu_star}, behaves asymptotically like a multiple of $\ones$. The best reaction from the investor's point of view is to diversify among all assets, weighted by the uncertainty structure $\Gamma$. In the special case where $K$ is a ball, this leads to a uniform diversification strategy. This result is in line with Pflug et al.~\cite{pflug_pichler_wozabal_2012} who show convergence of the optimal strategy to the uniform diversification strategy in a risk minimization setting with increasing model uncertainty.

\begin{remark}
	Note that plugging in $\kappa=\infty$ into the definition of the ellipsoid yields the uncertainty set $K=\R^d$, so that in fact every drift parameter $\mu\in\R^d$ is deemed possible by the investor. Then one easily obtains the worst-case utility
	\[ \inf_{\mu\in\R^d} \E_{\mu}[U_\gamma(X^\pi_T)]=
	\begin{cases}
		0,			& \gamma\in(0,1),\\
		-\infty,	& \gamma\in(-\infty,0],
	\end{cases} \]
	for \emph{any} admissible $\pi$. Hence, every strategy performs equally bad in the limit case. In particular, plugging in $\kappa=\infty$ into the ellipsoid in the first place does not provide us with the optimal limit strategy of Theorem~\ref{thm:limit_of_optimal_strategy}.
	
	The intuition is that, as long as the uncertainty set is bounded, there exists a worst-case drift to which the investor can react in an optimal way. Nevertheless, when uncertainty goes to infinity, also the expected utility achieved by the best strategy will be driven to $-\infty$ in case that $\gamma\in(-\infty,0]$, respectively to zero in case $\gamma\in(0,1)$.
\end{remark}

\subsection{Relaxing the investment constraint}\label{sec:relaxing_the_investment_constraint}

We use the above results to show that, as uncertainty $\kappa$ goes to infinity, our robust optimization problem yields the same optimal value as a slightly different optimization problem with a more general class of admissible strategies. Recall that we have so far considered for $h>0$ the set
\[ \calA_h(x_0)=\bigl\{ \pi\in\calA(x_0) \,\big|\, \langle\pi_t,\ones\rangle = h \text{ for all } t\in[0,T] \bigr\} \]
as the class of admissible strategies. Requiring $\langle \pi_t,\ones\rangle\geq h$ instead of $\langle \pi_t,\ones\rangle= h$ obviously enlarges this set. In the following, we show for logarithmic utility that maximizing worst-case expected utility among bounded strategies in this larger set asymptotically leads to the same value as our original problem. We write $K=K(\kappa)$ for the uncertainty ellipsoid with radius~$\kappa$.

\begin{proposition}\label{prop:comparison_greater_equal_h_logarithm}
	Define for $h>0$ the admissibility set
	\[ \calA'_h(x_0)=\bigl\{ \pi\in\calA(x_0) \,\big|\, \langle\pi_t,\ones\rangle \geq h \text{ for all } t\in[0,T] \bigr\} \]
	and let $M>0$. Then there exists a $\kappa_M>0$ such that for all $\kappa\geq\kappa_M$ it holds
	\[ \adjustlimits \sup_{\substack{\pi\in\calA'_h(x_0)\\ \lVert\pi\rVert\leq M}} \inf_{\substack{\mu\in K(\kappa)\\ \phantom{0}}} \E_\mu\bigl[\log(X^\pi_T)\bigr] \leq \adjustlimits \sup_{\pi\in\calA_h(x_0)} \inf_{\mu\in K(\kappa)} \E_\mu\bigl[\log(X^\pi_T)\bigr]. \]
	Here we use $\lVert\pi\rVert\leq M$ as a short notation for $\lVert\pi_t\rVert\leq M$ for all $t\in[0,T]$.
\end{proposition}

For power utility, the result is slightly weaker. We first give a lemma that states some useful equalities concerning the matrix $A$ and vector $c$ from Definition~\ref{def:matrix_A_vector_c}.

\begin{lemma}\label{lem:properties_of_A_and_c}
	For the matrix $A$ and the vector $c$ we have
	\[ A\sigma\sigma^\transp A=A, \quad c^\transp\sigma\sigma^\transp A=0 \quad\text{and}\quad c^\transp\ones =1. \]
\end{lemma}

The next proposition gives a result similar to Proposition~\ref{prop:comparison_greater_equal_h_logarithm} for power utility. We define a different enlarged admissibility set $\overline{\calA}_h(x_0)$ in this case. The reason is that, in contrast to the logarithmic utility case, we cannot ensure that we can restrict to deterministic strategies in $\calA'_h(x_0)$.

\begin{proposition}\label{prop:comparison_greater_equal_h_power}
	Let $\gamma\neq 0$ and $h>0$ and define the admissibility set
	\[ \overline{\calA}_h(x_0)=\bigcup_{h'\geq h} \calA_{h'}(x_0). \]
	Then there exists a $\kappa'>0$ such that for all $\kappa\geq\kappa'$ it holds
	\[ \adjustlimits \sup_{\pi\in \overline{\calA}_h(x_0)} \inf_{\mu\in K(\kappa)} \E_\mu\bigl[U_\gamma(X^\pi_T)\bigr] = \adjustlimits \sup_{\pi\in\calA_h(x_0)} \inf_{\mu\in K(\kappa)} \E_\mu\bigl[U_\gamma(X^\pi_T)\bigr]. \]
\end{proposition}

The previous propositions show that as uncertainty increases it is reasonable for investors to choose strategies $\pi$ with $\langle\pi_t,\ones\rangle$ as small as possible. Even if the class of admissible strategies is enlarged, the optimal value will for large uncertainty be attained by a strategy from $\calA_h(x_0)$. This is in line with the intuition from Proposition~\ref{prop:invest_only_in_bond}, where we have seen that as uncertainty exceeds a certain threshold, investors prefer to not invest anything into the risky assets.

\subsection{Risk aversion and speed of convergence}\label{sec:risk_aversion_and_speed_of_convergence}

As the class of admissible strategies we now take again
\[ \calA_h(x_0)=\bigl\{ \pi\in\calA(x_0) \,\big|\, \langle\pi_t,\ones\rangle = h \text{ for all } t\in[0,T] \bigr\} \]
for some $h>0$.
We have seen in Section~\ref{sec:limit_of_worst-case_parameter_and_optimal_strategy} that the optimal strategy $\pi^*(\kappa)$ for our robust optimization problem with ellipsoidal uncertainty sets $K$ converges as the level of uncertainty $\kappa$ goes to infinity. If the uncertainty set $K$ is a ball, then the limit is a uniform diversification strategy $\frac{h}{d}\ones$.
In the following, we illustrate this convergence by an example and investigate which influence the risk aversion parameter $\gamma$ has on the speed of convergence.
Note that for our class of utility functions, the value $1-\gamma$ is equal to the Arrow--Pratt measure of relative risk aversion. The smaller $\gamma$ is, the more risk-averse is the investor.

\begin{example}\label{ex:example_pi_star_against_kappa}
	We consider a market with $d=8$ risky assets. The volatility matrix has the form
	\[ \sigma=
	\begin{pmatrix*}[l]
		0.3	& 0	& 0	& 0	& 0	& 0	& 0	& 0 \\
		0.2	& 0.3	& 0	& 0	& 0	& 0	& 0	& 0 \\
		0	& 0.2	& 0.3	& 0	& 0	& 0	& 0	& 0 \\
		0.3	& 0.2	& 0	& 0.4	& 0	& 0	& 0	& 0 \\
		0.2	& 0.3	& 0	& 0.1	& 0.3	& 0	& 0	& 0 \\
		0.1	& 0.1	& 0.1	& 0.1	& 0.2	& 0.2	& 0	& 0 \\
		0.2	& 0.1	& 0.2	& 0.1	& 0.2	& 0.2	& 0.4	& 0 \\
		0.1	& 0	& 0	& 0.2	& 0.1	& 0.1	& 0.2	& 0.4
	\end{pmatrix*}. \]
	Investors use strategies from $\calA_h(x_0)$ with $h=1$. Further, we take $\Gamma=I_d$ and $\nu=\frac{3}{10}\ones$ as parameters of the uncertainty ellipsoid. Note that for this choice of the parameter $\nu$ the optimal strategy in the situation without model uncertainty, i.e.\ with $\kappa=0$, does not depend on $\gamma$.
	We then compute the constant optimal portfolio composition $\pi^*(\kappa)$ based on different values of $\gamma$ and for all $\kappa\in(0,0.5)$, and plot the result in Figure~\ref{fig:optimal_portfolio_composition_for_various_gamma} against $\kappa$.
	For any fixed level of uncertainty $\kappa$, the optimal composition $\pi^*(\kappa)$ is plotted as a stacked plot where every color corresponds to one stock.
	
	For small values of $\kappa$, the optimal strategy $\pi^*$ is negative in some components. This leads to an overall investment larger than one on the positive side. As $\kappa$ becomes larger, the composition gets closer and closer to the uniform diversification vector.
	When comparing the different subplots one sees that the convergence is faster for higher values of $\gamma$, an effect that has been shown to hold in general, see Westphal~\cite[Rem.~5.9]{westphal_2019}. This might be surprising at first glance since one expects a more risk-averse investor to choose a ``safer'' strategy sooner than a less risk-averse investor does. However, the effect becomes more intuitive when keeping in mind that we address a robust optimization problem where an investor is confronted with the worst possible drift parameter in the uncertainty set. An investor with a high, positive value of $\gamma$ would, in the non-robust problem, invest in the assets with the allegedly highest drift. In the worst-case market this undiversified strategy would allow the market to choose a very extreme drift parameter with high absolute values for exactly these assets. This implies that a less risk-averse investor is much more prone to the market's choice of a drift parameter. To make up for this, there is more diversification, which can even be amplified by the constraint using $h=1$, and thus the optimal robust strategy converges very fast, so that even for small values of uncertainty $\kappa$, the investor is already driven into the diversified uniform strategy.
	\begin{figure}[p]
		\centering
		\begin{subfigure}{.5\textwidth}
			\centering
			\setlength\figureheight{4.8cm}
			\setlength\figurewidth{0.85\textwidth}
			%\tikzsetnextfilename{speed_of_convergence_-2}
			%\input{figures/speed_of_convergence_-2.tikz}
			\includegraphics{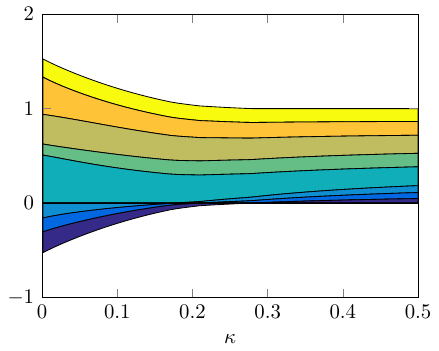}
			\caption{$\gamma=-2$}
		\end{subfigure}%
		\begin{subfigure}{.5\textwidth}
			\centering
			\setlength\figureheight{4.8cm}
			\setlength\figurewidth{0.85\textwidth}
			%\tikzsetnextfilename{speed_of_convergence_-1}
			%\input{figures/speed_of_convergence_-1.tikz}
			\includegraphics{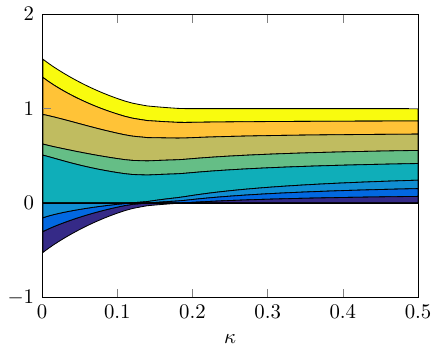}
			\caption{$\gamma=-1$}
		\end{subfigure}
		\newline
		\begin{subfigure}{.5\textwidth}
			\centering
			\setlength\figureheight{4.8cm}
			\setlength\figurewidth{0.85\textwidth}
			%\tikzsetnextfilename{speed_of_convergence_-0.5}
			%\input{figures/speed_of_convergence_-0.5.tikz}
			\includegraphics{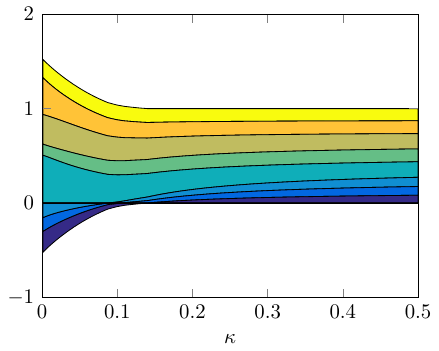}
			\caption{$\gamma=-0.5$}
		\end{subfigure}%
		\begin{subfigure}{.5\textwidth}
			\centering
			\setlength\figureheight{4.8cm}
			\setlength\figurewidth{0.85\textwidth}
			%\tikzsetnextfilename{speed_of_convergence_0}
			%\input{figures/speed_of_convergence_0.tikz}
			\includegraphics{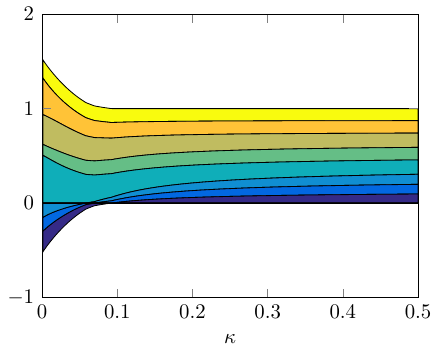}
			\caption{$\gamma=0$}
		\end{subfigure}
		\newline
		\begin{subfigure}{.5\textwidth}
			\centering
			\setlength\figureheight{4.8cm}
			\setlength\figurewidth{0.85\textwidth}
			%\tikzsetnextfilename{speed_of_convergence_0.5}
			%\input{figures/speed_of_convergence_0.5.tikz}
			\includegraphics{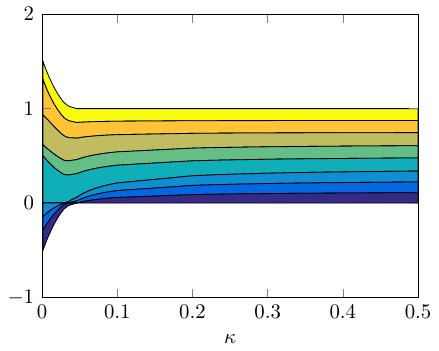}
			\caption{$\gamma=0.5$}
		\end{subfigure}%
		\begin{subfigure}{.5\textwidth}
			\centering
			\setlength\figureheight{4.8cm}
			\setlength\figurewidth{0.85\textwidth}
			%\tikzsetnextfilename{speed_of_convergence_0.9}
			%\input{figures/speed_of_convergence_0.9.tikz}
			\includegraphics{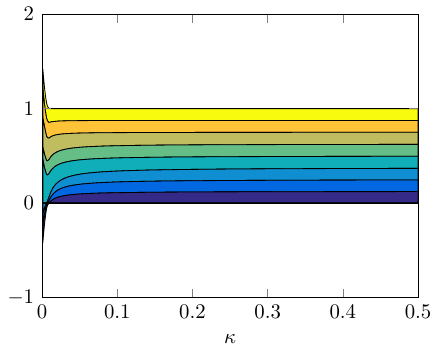}
			\caption{$\gamma=0.9$}
		\end{subfigure}
		\caption{Optimal portfolio composition $\pi^*$ plotted against $\kappa$ for different values of $\gamma$. The model parameters are given in Example~\ref{ex:example_pi_star_against_kappa}. For any $\gamma$, we observe convergence against a uniform diversification strategy. For larger values of $\gamma$, convergence appears to take place faster than for smaller values of $\gamma$.}\label{fig:optimal_portfolio_composition_for_various_gamma}
	\end{figure}
\end{example}

\subsection{Measures of robustness performance}\label{sec:measures_of_robustness_performance}

We have seen that introducing uncertainty in our utility maximization problem leads to more diversified strategies. The question arises what an investor gains from using robust strategies and what downside comes with behaving in a robust way in situations where it is not necessary. These two antithetic effects can be rated by the measures \emph{cost of ambiguity} and \emph{reward for distributional robustness} that have been studied in a different context in Analui~\cite[Sec.~3.4]{analui_2014}.

For our robust maximization problem, the center $\nu$ of the uncertainty ellipsoid can be seen as an estimation for the true drift of the stocks. If an investor was sure that the estimation was correct, she would simply maximize $\E_\nu[U_\gamma(X^\pi_T)]$. From Proposition~\ref{prop:optimal_strategy_non-robust} we know that the optimal strategy is then of the form $(\hat{\pi}_t)_{t\in[0,T]}$ with
\begin{equation}\label{eq:representation_of_pi_hat_robustness}
	\hat{\pi}_t = \frac{1}{1-\gamma}A\nu+hc
\end{equation}
for all $t\in[0,T]$. In the presence of uncertainty, the solution to our utility maximization problem is the strategy $(\pi^*_t)_{t\in[0,T]}$ with
\begin{equation}\label{eq:representation_of_pi_star_robustness}
	\pi^*_t = \frac{1}{1-\gamma}A\mu^*+hc
\end{equation}
for all $t\in[0,T]$, see Theorem~\ref{thm:duality_result}.
We now define measures for the robustness performance that consider the difference in the corresponding certainty equivalents when using $\hat{\pi}$ or $\pi^*$.

\begin{definition}
	We define the \emph{cost of ambiguity} as
	\[ \COA = U_\gamma^{-1}\bigl(\E_\nu\bigl[U_\gamma(X^{\hat{\pi}}_T)\bigr]\bigr)-U_\gamma^{-1}\bigl(\E_\nu\bigl[U_\gamma(X^{\pi^*}_T)\bigr]\bigr) \]
	and the \emph{reward for distributional robustness} as
	\[ \RDR = U_\gamma^{-1}\bigl(\E_{\mu^*}\bigl[U_\gamma(X^{\pi^*}_T)\bigr]\bigr)-U_\gamma^{-1}\bigl(\E_{\mu^*}\bigl[U_\gamma(X^{\hat{\pi}}_T)\bigr]\bigr). \]
\end{definition}

The cost of ambiguity captures how big the loss in the certainty equivalent is when using the robust strategy $\pi^*$, given that the estimation $\nu$ for the drift was actually correct. Note that $\hat{\pi}$ is the best strategy given drift $\nu$ and that $U_\gamma^{-1}$ is a strictly increasing function, hence $\COA$ is non-negative.
The reward for distributional robustness reflects how much an investor is rewarded when using the robust strategy $\pi^*$ compared to the ``naive'' strategy $\hat{\pi}$, assuming that indeed the worst possible drift parameter $\mu^*$ is the true one. We see that also $\RDR$ is non-negative since $\pi^*$ maximizes expected utility given $\mu^*$.

\begin{remark}
	A different definition of $\COA$ and $\RDR$ is possible where one measures the difference in expected utility rather than the difference of the certainty equivalents. The asymptotic behavior of the reward for distributional robustness for large uncertainty is then heavily affected by the parameter $\gamma$ of the investor's utility function. In particular, as $\kappa$ goes to infinity, the reward for distributional robustness goes to zero if $\gamma>0$ and to infinity if $\gamma<0$.
\end{remark}

\begin{proposition}\label{prop:oCOA_greater_than_oRDR}
	Independently of $\gamma\in(-\infty,1)$ it always holds $\COA\geq\RDR$.
\end{proposition}

Furthermore, $\COA$ and $\RDR$ converge as $\kappa$ goes to infinity. We write $\COA(\kappa)$ and $\RDR(\kappa)$ to emphasize the dependence on the degree of uncertainty.

\begin{proposition}\label{prop:limit_of_oCOA_and_oRDR}
	As $\kappa$ goes to infinity, $\COA(\kappa)$ converges to a non-negative limit and $\RDR(\kappa)$ goes to zero.
\end{proposition}

Figure~\ref{fig:study_robustness_ce} illustrates the behavior of $\COA$ and $\RDR$ in dependence on the level of uncertainty $\kappa$. We consider a market with $d=8$ stocks, where the underlying market parameters are those from Example~\ref{ex:example_pi_star_against_kappa}. The figure shows $\COA$ and $\RDR$ plotted against $\kappa$ for different values of $\gamma$. Note that the scaling in the second row of subfigures is different from the scaling in the first row. The absolute values of $\COA$ and $\RDR$ become smaller as $\gamma$ increases.

We observe that the qualitative behavior of $\COA$ and $\RDR$ is the same for any value of the risk aversion coefficient $\gamma$. For any fixed $\gamma$ and $\kappa$, $\RDR$ is always less than $\COA$, a property that we have proven in Proposition~\ref{prop:oCOA_greater_than_oRDR}. As $\kappa$ increases, $\COA$ goes to a finite positive limit, whereas $\RDR$ tends to zero, as we have shown in Proposition~\ref{prop:limit_of_oCOA_and_oRDR}.

\begin{figure}[ht]
	\begin{subfigure}{.29\textwidth}
		\centering
		\setlength\figureheight{3cm}
		\setlength\figurewidth{0.82\textwidth}
		%\tikzsetnextfilename{robustness_measures_ce_-1}
		%\input{figures/robustness_measures_ce_-1.tikz}
		\includegraphics{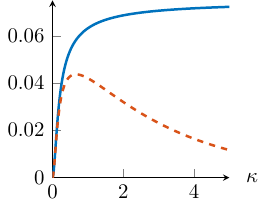}
		\setlength{\abovecaptionskip}{-10pt}
		\setlength{\belowcaptionskip}{10pt}
		\caption{$\gamma=-1$}
	\end{subfigure}%
	\begin{subfigure}{.29\textwidth}
		\centering
		\setlength\figureheight{3cm}
		\setlength\figurewidth{0.82\textwidth}
		%\tikzsetnextfilename{robustness_measures_ce_-0.5}
		%\input{figures/robustness_measures_ce_-0.5.tikz}
		\includegraphics{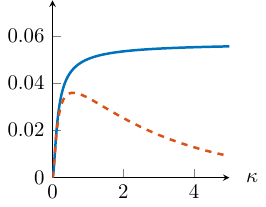}
		\setlength{\abovecaptionskip}{-10pt}
		\setlength{\belowcaptionskip}{10pt}
		\caption{$\gamma=-0.5$}
	\end{subfigure}%
	\begin{subfigure}{.29\textwidth}
		\centering
		\setlength\figureheight{3cm}
		\setlength\figurewidth{0.82\textwidth}
		%\tikzsetnextfilename{robustness_measures_ce_-0.1}
		%\input{figures/robustness_measures_ce_-0.1.tikz}
		\includegraphics{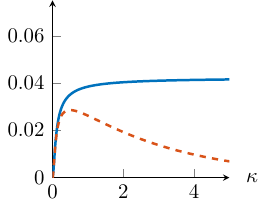}
		\setlength{\abovecaptionskip}{-10pt}
		\setlength{\belowcaptionskip}{10pt}
		\caption{$\gamma=-0.1$}
	\end{subfigure}
	\newline
	\begin{subfigure}{.29\textwidth}
		\centering
		\setlength\figureheight{3cm}
		\setlength\figurewidth{0.82\textwidth}
		%\tikzsetnextfilename{robustness_measures_ce_0}
		%\input{figures/robustness_measures_ce_0.tikz}
		\includegraphics{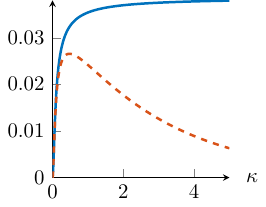}
		\caption{$\gamma=0$}
	\end{subfigure}%
	\begin{subfigure}{.29\textwidth}
		\centering
		\setlength\figureheight{3cm}
		\setlength\figurewidth{0.82\textwidth}
		%\tikzsetnextfilename{robustness_measures_ce_0.1}
		%\input{figures/robustness_measures_ce_0.1.tikz}
		\includegraphics{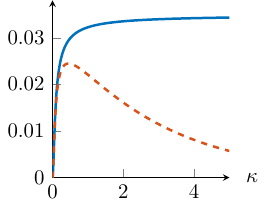}
		\caption{$\gamma=0.1$}
	\end{subfigure}%
	\begin{subfigure}{.29\textwidth}
		\centering
		\setlength\figureheight{3cm}
		\setlength\figurewidth{0.82\textwidth}
		%\tikzsetnextfilename{robustness_measures_ce_0.5}
		%\input{figures/robustness_measures_ce_0.5.tikz}
		\includegraphics{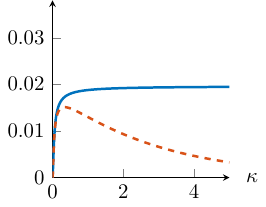}
		\caption{$\gamma=0.5$}
	\end{subfigure}%
	\begin{subfigure}{.13\textwidth}
		\centering
		\setlength\figureheight{3cm}
		\setlength\figurewidth{0.5\textwidth}
		%		\tikzsetnextfilename{legend_COA_RDR_ce}
		%		\begin{tikzpicture}
		%			\begin{axis}[hide axis, xmin=0, xmax=1, ymin=-1, ymax=1,
		%			legend style={legend cell align=left, anchor=north, at={(0,1)}}, font=\small]
		%				
		%				\definecolor{matlabblue}{rgb}{0.00000,0.44700,0.74100}
		%				\definecolor{matlabred}{rgb}{0.85000,0.32500,0.09800}
		%				
		%				\addlegendimage{empty legend}
		%				\addlegendentry{\hspace{-.7cm}\textbf{Legend}};
		%				\addlegendimage{color=matlabblue,solid,very thick}
		%				\addlegendentry{$\COA$};
		%				\addlegendimage{color=matlabred,dashed,very thick}
		%				\addlegendentry{$\RDR$};
		%			\end{axis}
		%		\end{tikzpicture}
		\includegraphics{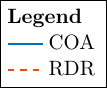}
	\end{subfigure}
	\caption{The behavior of $\COA$ and $\RDR$ plotted against uncertainty radius $\kappa$ for different values of the risk aversion coefficient $\gamma$. The parameters are those from Example~\ref{ex:example_pi_star_against_kappa}.}\label{fig:study_robustness_ce}
\end{figure}

\section{Outlook on stochastic drift and time-dependent uncertainty sets}\label{sec:outlook}

In this section we want to give a brief outlook on how the results of this paper can be applied also in more general financial market models with a stochastic drift process. This generalization is the topic of our follow-up work Sass and Westphal~\cite{sass_westphal_2021}. Here we only give a short outline of the setup to illustrate the relevance of this work.

In Sass and Westphal~\cite{sass_westphal_2021} the results of the present paper are generalized to a financial market with a stochastic drift process and time-dependent uncertainty sets $(K_t)_{t\in[0,T]}$. This is motivated by the idea that information about the hidden drift process, as e.g.\ obtained from filtering techniques, might change over time. A surplus of information should then be reflected in a smaller uncertainty set.
More precisely, we assume that under the reference measure returns follow the dynamics
\[ \rmd R_t = \nu_t\,\rmd t+\sigma\,\rmd W_t, \]
where the reference drift $(\nu_t)_{t\in[0,T]}$ is adapted to the filtration $(\calG_t)_{t\in[0,T]}$ representing the investor's information. This is justified by a separation principle where one performs a filtering step before solving the optimization problem, i.e.\ $(\nu_t)_{t\in[0,T]}$ represents the investor's filter for the drift process. We introduce a time-dependent uncertainty set $(K_t)_{t\in[0,T]}$ that is a set-valued stochastic process adapted to $(\calG_t)_{t\in[0,T]}$, meaning that the investor knows the realization of $K_t$ at time $t$.

It is not obvious how to set up a worst-case optimization problem in this time-dependent setting. The problem lies in the fact that the realization of the uncertainty sets $(K_t)_{t\in[0,T]}$ is not known initially but gets revealed over time. A worst-case drift process $(\mu_t)_{t\in[0,T]}$ is characterized by being the worst one with the property that $\mu_t\in K_t$ for all $t\in[0,T]$. However, optimization with respect to this worst-case drift process is not feasible for an investor since it is not known initially. Instead, it makes sense to consider the following local approach. For any fixed $t\in[0,T]$, the current uncertainty set $K_t$ is known. Given this $K_t$, investors take model uncertainty into account by assuming that in the future the worst possible drift process having values in $K_t$ will be realized, i.e.\ the worst drift process from the class
\[ \calK^{(t)}=\bigl\{(\mu^{(t)}_s)_{s\in[t,T]} \,\big|\, \mu^{(t)}_s\in K_t \text{ and } \mu^{(t)}_s \text{ is }\calG_t\text{-measurable for each }s\in[t,T]\bigr\}. \]
Investors then solve at each time $t\in[0,T]$ the local optimization problem
\begin{equation}\label{eq:optimization_problem_time_dependent}
	\adjustlimits \sup_{\pi^{(t)}\in\calA_h(t,x)} \inf_{\mu^{(t)}\in \calK^{(t)}} \E_{\mu^{(t)}}\Bigl[U\bigl(X^{t,x,\pi^{(t)}}_T\bigr)\Bigr].
\end{equation}
Here, we write $X^{t,x,\pi^{(t)}}_s$ for the wealth at $s\in[t,T]$ when starting at~$t$ with~$x$ and using strategy $\pi^{(t)}\in\calA_h(t,x)$, where the admissibility set is defined analogously to $\calA_h(x_0)$ for strategies starting at $t$. This leads to an optimal strategy $(\pi^{(t),*}_s)_{s\in[t,T]}$. In our continuous-time setting this decision will be revised as soon as $K_t$ changes, possibly continuously in time. The realized optimal strategy of the investor is then given by $\pi^*_t=\pi^{(t),*}_t$ for any $t\in[0,T]$.

This setup of the local optimization problems is reasonable from an investor's point of view. The uncertainty sets $K_t$ change continuously in time due to new incoming information along with return observations, for example. Naturally, the optimal strategy of the investor will then also be adapted continuously. In Sass and Westphal~\cite{sass_westphal_2021} it is shown in detail how the results of this paper can be used to solve the above described more complicated problem. An explicit representation of the optimal strategy and a minimax theorem can be derived. Those results then also apply to much more general financial market models.
The convergence results from Section~\ref{cha:asymptotic_behavior_as_uncertainty_increases}, however, do not have a straightforward analogon in the setting with time-dependent uncertainty sets.

\begin{remark}	
	Initially it is not clear whether we have an inconsistent control problem, cf.\ Bj\"{o}rk et al.~\cite{bjoerk_khapko_murgoci_2017}, in our original formulation~\eqref{eq:robust_problem_power_log}. But for the special case of~\eqref{eq:optimization_problem_time_dependent} with a constant uncertainty set $K$, the results in Sass and Westphal~\cite{sass_westphal_2021} show that one obtains at time~$t$ the same optimal risky fractions as when starting at time~$0$. In combination with the Bellman principle, which implies that at time~$t$ we only need the information $X_t=x$, this proves that our robust utility maximization problem with optimal solution $\pi^*$ obtained in Section~\ref{cha:a_duality_approach} is time-consistent.
	A generalization to allowing for more probability measures than those corresponding to a constant drift in a formulation based on a robust utility functional may raise consistency issues and would need assumptions on the structure of this set. This may then be treated as in M\"{u}ller~\cite{mueller_2005} under appropriate conditions.
\end{remark}

\appendix

\section{Proofs}\label{app:proofs}

For better readability of the paper, all proofs are collected in this appendix.

\begin{proof}[Proof of Proposition~\ref{prop:invest_only_in_bond}]
	Let $\mu\in K$ and $\pi\in\calA(x_0)$.
	We consider the case $\gamma=0$ first. The expected logarithmic utility of terminal wealth under measure $\mathbb{P}^\mu$ is
	\[ \E_\mu\bigl[\log(X^\pi_T)\bigr] = \log(x_0)+\E_\mu\biggl[\int_0^T \Bigl(r+\pi_t^\transp(\mu-r\ones)-\frac{1}{2}\lVert\sigma^\transp\pi_t\rVert^2\Bigr)\rmd t\biggr]. \]
	Since the vector $r\ones$ is an element of the set $K$, we immediately see that
	\[ \inf_{\mu\in K} \E_\mu\bigl[\log(X^\pi_T)\bigr] \leq \E_{r\ones}\bigl[\log(X^\pi_T)\bigr] \leq \log(x_0)+rT, \]
	so we can deduce that the trivial strategy $\pi\equiv 0$ is optimal for~\eqref{eq:recall_value_function_robust}, since $\pi\equiv 0$ leads to expected utility of terminal wealth $\log(x_0)+rT$ under each of the measures $\mathbb{P}^\mu$.
	
	For power utility, i.e.\ $\gamma\neq 0$, the argumentation is similar. Since $r\ones\in K$, we have
	\[ \inf_{\mu\in K} \E_\mu\bigl[U_\gamma(X^\pi_T)\bigr] \leq \frac{x_0^\gamma}{\gamma}\rme^{\gamma rT}\E_{r\ones}\biggl[ \exp\biggl(-\frac{\gamma}{2}\int_0^T\lVert\sigma^\transp\pi_t\rVert^2\,\rmd t+ \gamma\int_0^T \pi_t^\transp \sigma\,\rmd W^{r\ones}_t\biggr) \biggr] \]
	and we can rewrite the expectation on the right-hand side as
	\[ \E_{r\ones}\biggl[ \exp\biggl(\gamma\int_0^T\! \pi_t^\transp \sigma\,\rmd W^{r\ones}_t-\frac{1}{2}\gamma^2\int_0^T\!\lVert\sigma^\transp\pi_t\rVert^2\,\rmd t\biggr)\exp\biggl( -\frac{1}{2}\gamma(1-\gamma)\int_0^T\!\lVert\sigma^\transp\pi_t\rVert^2\,\rmd t \biggr) \biggr]. \]
	Thus,
	\[ \inf_{\mu\in K} \E_\mu\bigl[U_\gamma(X^\pi_T)\bigr] \leq \frac{x_0^\gamma}{\gamma}\rme^{\gamma rT}\E_{r\ones}\biggl[ \exp\biggl(\gamma\int_0^T \pi_t^\transp \sigma\,\rmd W^{r\ones}_t-\frac{1}{2}\gamma^2\int_0^T\lVert\sigma^\transp\pi_t\rVert^2\,\rmd t\biggr)\biggr]. \]
	But the exponential local martingale in the expression above has expectation less or equal than one, so
	\[ \inf_{\mu\in K} \E_\mu\bigl[U_\gamma(X^\pi_T)\bigr] \leq \frac{x_0^\gamma}{\gamma}\rme^{\gamma rT}. \]
	Again, as for logarithmic utility, the trivial strategy $\pi\equiv 0$ is optimal for~\eqref{eq:recall_value_function_robust} if $r\ones\in K$, since the zero strategy leads exactly to expected power utility $\frac{x_0^\gamma}{\gamma}\rme^{\gamma rT}$.
\end{proof}

\begin{proof}[Proof of Lemma~\ref{lem:definition_D}]
	Since $d\leq m$ and $\sigma\in\R^{d\times m}$ has rank $d$, the rows of $\sigma$ are independent vectors in $\R^m$. Now $D\sigma\in\R^{(d-1)\times m}$ and due to the specific form of $D$, the $i$-th row of $D\sigma$ is $\sigma_{i,\cdot}-\sigma_{d,\cdot}$, $i=1,\dots,d-1$. Here, $\sigma_{i,\cdot}$ denotes the $i$-th row of matrix $\sigma$.
	Now from the independence of $\sigma_{1,\cdot},\dots,\sigma_{d,\cdot}$ it follows for any $a_1,\dots,a_{d-1}\in\R$ that if
	\[ 0=\sum_{i=1}^{d-1} a_i(\sigma_{i,\cdot}-\sigma_{d,\cdot})=\sum_{i=1}^{d-1} a_i\sigma_{i,\cdot}-\sum_{i=1}^{d-1} a_i\sigma_{d,\cdot},  \]
	then $a_1=\cdots=a_{d-1}=0$. Hence, the rows of $D\sigma$ are independent, and $\mathrm{rank}(D\sigma)=d-1$.
\end{proof}

\begin{proof}[Proof of Proposition~\ref{prop:optimal_strategy_non-robust}]
	Let $\pi\in\calA_h(x_0)$. Then $\pi^d_t = h-\sum_{i=1}^{d-1}\pi^i_t$ for all $t\in[0,T]$, therefore we can transform
	\begin{align}
		\pi_t^\transp(\mu-r\ones) &= \sum_{i=1}^{d-1}\pi^i_t(\mu^i-r) + \biggl(h-\sum_{i=1}^{d-1}\pi^i_t\biggr)(\mu^d-r) = h(e_d^\transp\mu-r)+\widetilde{\pi}_t^\transp D\mu,\label{eq:wealth_representation_Lebesgue_term}\\
		\pi_t^\transp\sigma &= \sum_{i=1}^{d-1} \pi^i_t\sigma_{i,\cdot} + \biggl(h-\sum_{i=1}^{d-1}\pi^i_t\biggr)\sigma_{d,\cdot} = he_d^\transp\sigma+\widetilde{\pi}_t^\transp D\sigma,\label{eq:wealth_representation_stochastic_term}
	\end{align}
	where $\widetilde{\pi}_t:=\pi^{1:d-1}_t$ for all $t\in[0,T]$, and where $\sigma_{i,\cdot}$ denotes the $i$-th row of matrix $\sigma$.
	In the representation of the wealth process we first plug in~\eqref{eq:wealth_representation_stochastic_term} into the stochastic integral.
	For $\gamma\neq 0$ we perform a change of measure
	\[ \frac{\rmd \widetilde{\mathbb{P}}}{\rmd \mathbb{P}^\mu}=Z_T=\exp\biggl( \int_0^T \gamma he_d^\transp\sigma\,\rmd W^\mu_t -\frac{1}{2}\int_0^T \lVert \gamma h \sigma^\transp e_d \rVert^2\,\rmd t \biggr), \]
	such that under the measure $\widetilde{\mathbb{P}}$, the process $(\widetilde{W}^\mu_t)_{t\in[0,T]}$ with $\widetilde{W}^\mu_t=W^\mu_t-\int_0^t \gamma h\sigma^\transp e_d\,\rmd s$ is a Brownian motion by Girsanov's Theorem. We obtain
	\begin{equation*}
		\begin{aligned}
			&\E_\mu\biggl[\exp\biggl(\gamma\int_0^T \Bigl(\pi_t^\transp(\mu-r\ones) -\frac{1}{2}\lVert\sigma^\transp\pi_t\rVert^2\Bigr)\rmd t + \gamma\int_0^T \pi_t^\transp \sigma\,\rmd W^\mu_t \biggr)\biggr]\\
			&= \widetilde{\E}\biggl[\exp\biggl(\gamma\!\int_0^T\!\!\! \Bigl(\pi_t^\transp(\mu-r\ones) -\frac{1}{2}\lVert\sigma^\transp\!\pi_t\rVert^2+\frac{1}{2}\gamma\lVert h \sigma^\transp e_d \rVert^2\Bigr)\rmd t +\!\int_0^T \!\!\!\gamma\widetilde{\pi}_t^\transp D\sigma\,\rmd W^\mu_t \biggr)\biggr]\\
			&= \widetilde{\E}\biggl[\exp\biggl( \gamma\int_0^T \!\!\Bigl(\pi_t^\transp(\mu-r\ones) -\frac{1}{2}\lVert\sigma^\transp\pi_t\rVert^2+\frac{1}{2}\gamma\lVert h \sigma^\transp e_d \rVert^2+\gamma h\widetilde{\pi}_t^\transp D\sigma\sigma^\transp e_d\Bigr)\rmd t \\
			&\qquad\qquad\qquad+\int_0^T \!\!\gamma\widetilde{\pi}_t^\transp D\sigma\,\rmd \widetilde{W}^\mu_t \biggr)\biggr].
		\end{aligned}
	\end{equation*}
	By straightforward calculations using~\eqref{eq:wealth_representation_Lebesgue_term} and~\eqref{eq:wealth_representation_stochastic_term} the integrand in the Lebesgue integral above can be rewritten as
	\begin{equation*}
		\widetilde{\pi}_t^\transp\bigl(D\mu - h(1-\gamma)D\sigma\sigma^\transp e_d\bigr)-\frac{1}{2}\lVert(D\sigma)^\transp \widetilde{\pi}_t \rVert^2 +he_d^\transp\mu-hr-\frac{1}{2}(1-\gamma)\lVert h\sigma^\transp e_d \rVert^2.
	\end{equation*}
	If we now substitute
	\begin{equation}\label{eq:substitution_r_mu_sigma}
		\begin{aligned}
			\widetilde{\sigma}&=D\sigma, \\
			\widetilde{r}&=(1-h)r+he_d^\transp\mu-\frac{1}{2}(1-\gamma)\lVert h\sigma^\transp e_d \rVert^2, \\
			\widetilde{\mu}&=D\mu - h(1-\gamma)D\sigma\sigma^\transp e_d+\widetilde{r}\mathbf{1}_{d-1},
		\end{aligned}
	\end{equation}
	then the expected utility of terminal wealth is given by
	\begin{equation}\label{eq:reduction_of_utility_to_d-1}
		\begin{aligned}
			&\E_\mu\bigl[ U_\gamma(X^\pi_T) \bigr] =\frac{x_0^\gamma}{\gamma}\widetilde{\E}\biggl[ \exp\biggl(\gamma\int_0^T \!\Bigl(\widetilde{r}+\widetilde{\pi}_t^\transp(\widetilde{\mu}-\widetilde{r}\mathbf{1}_{d-1})-\frac{1}{2}\lVert\widetilde{\sigma}^\transp\widetilde{\pi}_t\rVert^2\Bigr) \rmd t +\gamma\int_0^T \!\widetilde{\pi}_t^\transp \widetilde{\sigma}\,\rmd \widetilde{W}_t \biggr) \biggr].
		\end{aligned}
	\end{equation}
	In the case $\gamma=0$, like in the power utility case, we can represent expected utility of terminal wealth as
	\begin{equation}\label{eq:reduction_of_utility_to_d-1_log}
		\begin{aligned}
			\E_\mu\bigl[\log(X^\pi_T)\bigr] &= \log(x_0)+\widetilde{r}\,T+\E\biggl[\int_0^T \Bigl(\widetilde{\pi}_t^\transp \bigl(\widetilde{\mu}-\widetilde{r}\mathbf{1}_{d-1}\bigr) -\frac{1}{2}\lVert\widetilde{\sigma}^\transp \widetilde{\pi}_t \rVert^2\Bigr)\rmd t\biggr],
		\end{aligned}
	\end{equation}
	where we use the same substitution with $\widetilde{r}$, $\widetilde{\mu}$ and $\widetilde{\sigma}$ as in~\eqref{eq:substitution_r_mu_sigma} for $\gamma=0$.
	
	In both cases $\gamma\neq 0$ and $\gamma=0$ we realize that the expressions in~\eqref{eq:reduction_of_utility_to_d-1} and~\eqref{eq:reduction_of_utility_to_d-1_log} are again the expected utility of terminal wealth in a financial market with $d-1$ risky assets where the risk-free interest rate is $\widetilde{r}$, the drift of the $d-1$ risky assets is given by $\widetilde{\mu}\in\R^{d-1}$, and the volatility matrix is $\widetilde{\sigma}\in\R^{(d-1)\times m}$.
	So we have reduced the $d$-dimensional constrained problem to a $(d-1)$-dimensional unconstrained problem. When trying to maximize the right-hand side of~\eqref{eq:reduction_of_utility_to_d-1}, respectively~\eqref{eq:reduction_of_utility_to_d-1_log}, over all admissible strategies $\widetilde{\pi}$ with values in $\R^{d-1}$, we know from Merton~\cite{merton_1969} that the optimal strategy is constant in time and has the form
	\begin{equation*}
		\widetilde{\pi}_t = \frac{1}{1-\gamma}(\widetilde{\sigma}\widetilde{\sigma}^\transp)^{-1}(\widetilde{\mu}-\widetilde{r}\mathbf{1}_{d-1})
		= \frac{1}{1-\gamma}(D\sigma\sigma^\transp D^\transp)^{-1}\bigl(D\mu - h(1-\gamma)D\sigma\sigma^\transp e_d\bigr).
	\end{equation*}
	We now return to our original $d$-dimensional market, using the relation $\pi_t=D^\transp\widetilde{\pi}_t+he_d$, giving us the optimal strategy $\pi$ with
	\[ \pi_t = \frac{1}{1-\gamma}D^\transp(D\sigma\sigma^\transp D^\transp)^{-1}D\mu +h\bigl(I_d-D^\transp(D\sigma\sigma^\transp D^\transp)^{-1}D\sigma\sigma^\transp\bigr)e_d=\frac{1}{1-\gamma}A\mu+hc.\qedhere \]
\end{proof}

\begin{proof}[Proof of Lemma~\ref{lem:A_symmetric_positive_semidefinite}]
	Note that $D\sigma\sigma^\transp D^\transp$ is symmetric. Hence, the same is true for its inverse and thus for $D^\transp(D\sigma\sigma^\transp D^\transp)^{-1}D$.
	Also, $D\sigma\sigma^\transp D^\transp=(D\sigma)(D\sigma)^\transp$ is positive definite since $\sigma\in\R^{d\times m}$ has rank $d$ and therefore by Lemma~\ref{lem:definition_D}, $D\sigma$ has full row rank $d-1$. It follows that also the inverse $(D\sigma\sigma^\transp D^\transp)^{-1}$ is positive definite. So since
	\[ x^\transp Ax=x^\transp D^\transp(D\sigma\sigma^\transp D^\transp)^{-1}Dx=(Dx)^\transp(D\sigma\sigma^\transp D^\transp)^{-1}(Dx)\geq 0 \]
	for any $x\in\R^d$, the matrix $A$ is positive semidefinite.
	Furthermore, it is easy to check that $\mathrm{ker}(D)=\mathrm{span}(\{\ones\})$ and $\mathrm{ker}(D^\transp)=\{0\}$. Hence, it holds $Ax=D^\transp(D\sigma\sigma^\transp D^\transp)^{-1}Dx=0$ if and only if $(D\sigma\sigma^\transp D^\transp)^{-1}Dx=0$, which is equivalent to $Dx=0$. Hence we can deduce $\mathrm{ker}(A)=\mathrm{ker}(D)=\mathrm{span}(\{\ones\})$.
\end{proof}

\begin{proof}[Proof of Lemma~\ref{lem:worst-case_parameter}]
	Recall that $\tau^\transp A\tau$ has an eigenvalue $\lambda_1=0$ with a corresponding normed eigenvector of the form $v_1=\frac{1}{\lVert \tau^{-1}\ones \rVert}\tau^{-1}\ones$. The other eigenvalues of $\tau^\transp A\tau$ are positive and due to symmetry we can assume that $v_1,\dots,v_d$ form an orthogonal basis of $\R^d$.
	Firstly, we observe that the gradient of $g$ is
	\[ \nabla g(\rho) = \frac{1}{2(1-\gamma)}2\tau^\transp A\tau\rho+\tau^\transp\Bigl(hc+\frac{1}{1-\gamma}A\nu\Bigr)
	= \tau^\transp\biggl( A\Bigl(\frac{1}{1-\gamma}(\tau\rho+\nu)-h\sigma\sigma^\transp e_d\Bigr)+h e_d\biggr). \]
	It follows that there is no $\rho\in B_\kappa(0)$ with $\nabla g(\rho)=0$, since $\tau^\transp$ is nonsingular and $he_d$ is not in the range of $A=D^\transp(D\sigma\sigma^\transp D^\transp)^{-1}D$. The minimum of $g$ on $B_\kappa(0)$ is therefore attained on the boundary.
	
	Let $\rho\in B_\kappa(0)$ be arbitrary. Since $v_1,\dots,v_d$ form a basis of $\R^d$, we are able to write $\rho=\sum_{i=1}^d a_iv_i$, where $a_1,\dots,a_d\in\R$ are uniquely determined. Since we know that a minimizer of the function $g$ must lie on the boundary of $B_\kappa(0)$ we obtain the constraint
	\begin{equation}\label{eq:norm_of_mu_squared}
		\kappa^2=\lVert\rho\rVert^2=\sum_{i=1}^d a_i^2
	\end{equation}
	on the coefficients. Before doing the minimization, we first notice that for our minimizer, the coefficient $a_1$ will be less or equal than zero. This is because
	\begin{equation}\label{eq:function_g_influence_a_1}
		\begin{aligned}
			g\biggl(\sum_{i=1}^d a_iv_i\biggr)
			&= \frac{1}{2(1-\gamma)}\biggl(\sum_{i=1}^d a_iv_i\biggr)^\transp \tau^\transp A\tau\biggl(\sum_{i=1}^d a_iv_i\biggr)+\Bigl(hc+\frac{1}{1-\gamma}A\nu\Bigr)^\transp\tau\biggl(\sum_{i=1}^d a_iv_i\biggr) \\
			%&= \frac{1}{2(1-\gamma)}\sum_{i=1}^d\sum_{j=1}^d a_ia_jv_i^\transp \tau^\transp A\tau v_j +\sum_{i=1}^d a_ihc^\transp\tau v_i+\frac{1}{1-\gamma}\sum_{i=1}^d a_i(A\nu)^\transp\tau v_i \\
			&= \frac{1}{2(1-\gamma)}\sum_{i=1}^d a_i^2\lambda_i+\sum_{i=1}^d a_ihc^\transp\tau v_i +\frac{1}{1-\gamma}\sum_{i=1}^d a_i\nu^\transp\lambda_i(\tau^\transp)^{-1} v_i \\
			&= \frac{1}{2(1-\gamma)}\sum_{i=2}^d a_i^2\lambda_i+\sum_{i=2}^d a_i\Bigl(hc+\frac{\lambda_i}{1-\gamma}\Gamma^{-1}\nu\Bigr)^\transp\tau v_i +a_1hc^\transp\tau v_1.
		\end{aligned}
	\end{equation}
	Next, one easily sees that
	\begin{equation}\label{eq:c_transposed_v_1}
		c^\transp\tau v_1=e_d^\transp(I_d-A\sigma\sigma^\transp)^\transp\tau\frac{1}{\lVert \tau^{-1}\ones \rVert}\tau^{-1}\ones=\frac{1}{\lVert \tau^{-1}\ones \rVert}e_d^\transp(\ones-\sigma\sigma^\transp A\ones)=\frac{1}{\lVert \tau^{-1}\ones \rVert},
	\end{equation}
	since $A\ones=0$ by Lemma~\ref{lem:A_symmetric_positive_semidefinite}. By plugging in this representation we deduce that, when looking for the minimizer of $g$, we can restrict to the parameters $\rho$ with coefficient $a_1\leq 0$. 
	We obtain
	\begin{equation*}
		\begin{aligned}
			\widetilde{g}(a_2,\dots,&a_d)
			:=g\biggl(\sum_{i=1}^d a_iv_i\biggr) \\
			&= \frac{1}{2(1-\gamma)}\sum_{i=2}^d a_i^2\lambda_i+\sum_{i=2}^d a_i\Bigl(hc+\frac{\lambda_i}{1-\gamma}\Gamma^{-1}\nu\Bigr)^\transp\tau v_i -\frac{h}{\lVert \tau^{-1}\ones \rVert}\sqrt{\kappa^2-\sum_{i=2}^d a_i^2},
		\end{aligned}
	\end{equation*}
	and minimize this expression in $a_2,\dots,a_d$.
	Note that the domain of $\widetilde{g}$ is $\{x\in\R^{d-1}\,|\,\lVert x\rVert\leq\kappa\}$. In the interior of this domain, the partial derivative of $\widetilde{g}$ with respect to $a_k$, $k=2,\dots,d$, is given by
	\begin{equation*}
		\begin{aligned}
			\frac{\partial \widetilde{g}}{\partial a_k}(a_2,\dots,a_d)
			&= \frac{2a_k\lambda_k}{2(1-\gamma)}+\Bigl(hc+\frac{\lambda_k}{1-\gamma}\Gamma^{-1}\nu\Bigr)^\transp \tau v_k-\frac{h}{2\lVert \tau^{-1}\ones \rVert\sqrt{\kappa^2-\sum_{i=2}^d a_i^2}}(-2a_k) \\
			&= \Biggl(\frac{\lambda_k}{1-\gamma}+\frac{h}{\lVert \tau^{-1}\ones \rVert\sqrt{\kappa^2-\sum_{i=2}^d a_i^2}}\Biggr)a_k+\Bigl(hc+\frac{\lambda_k}{1-\gamma}\Gamma^{-1}\nu\Bigr)^\transp\tau v_k.
		\end{aligned}
	\end{equation*}
	When setting this expression equal to zero, we obtain
	\begin{equation}\label{eq:the_coefficients_a_k}
		\begin{aligned}
			a_k %&= -\Biggl( \frac{\lambda_k}{1-\gamma}+\frac{h}{\lVert \tau^{-1}\ones \rVert\sqrt{\kappa^2-\sum_{i=2}^d a_i^2}} \Biggr)^{-1}\Bigl(hc+\frac{\lambda_k}{1-\gamma}\Gamma^{-1}\nu\Bigr)^\transp\tau v_k \\
			&= -\biggl( \frac{\lambda_k}{1-\gamma}-\frac{h}{\lVert \tau^{-1}\ones \rVert a_1} \biggr)^{-1}\biggl\langle h\tau^\transp c+\frac{\lambda_k}{1-\gamma}\tau^{-1}\nu, v_k\biggr\rangle.
		\end{aligned}
	\end{equation}
	Note that this representation does not provide the coefficients $a_k$ explicitly since $a_1$ here is a function of $(a_2,\dots,a_d)$.
	However, it is easy to check that the function
	\[ [-\kappa,0) \ni a_1 \mapsto a_1^2+\sum_{i=2}^d \biggl( \frac{\lambda_i}{1-\gamma}-\frac{h}{\lVert \tau^{-1}\ones \rVert a_1} \biggr)^{-2}\biggl\langle h\tau^\transp c+\frac{\lambda_i}{1-\gamma}\tau^{-1}\nu, v_i\biggr\rangle^2 \]
	has a strictly negative derivative on $[-\kappa,0)$. For $a_1=-\kappa$, the value of the function is greater or equal $\kappa^2$, for $a_1$ tending to zero from below it converges to zero, hence there is a unique value of $a_1\in[-\kappa,0)$ where the function has value $\kappa^2$. So~\eqref{eq:the_coefficients_a_k} together with~\eqref{eq:norm_of_mu_squared} uniquely determines $a_1,\dots,a_d$.
	
	Moreover, by some straightforward calculations we see that the Hessian of $\widetilde{g}$ is of the form
	\begin{align*}
		\frac{1}{1-\gamma}\widetilde{\Lambda}&+\frac{h}{\lVert \tau^{-1}\ones \rVert\sqrt{\kappa^2-\sum_{i=2}^d a_i^2}}I_{d-1}\\
		&+\frac{h}{\lVert \tau^{-1}\ones \rVert\bigl(\kappa^2-\sum_{i=2}^d a_i^2\bigr)^{3/2}}(a_2,\dots,a_d)^\transp(a_2,\dots,a_d),
	\end{align*}
	where $\widetilde{\Lambda}\in\R^{(d-1)\times(d-1)}$ is a diagonal matrix with diagonal entries $\lambda_2,\dots,\lambda_d>0$. Obviously, the first two summands are positive-definite matrices. The last summand is positive semidefinite. So we conclude that the Hessian of $\widetilde{g}$ is positive definite on the whole interior of the domain of $\widetilde{g}$. In particular, in the point $(a_2,\dots,a_d)$ defined via~\eqref{eq:the_coefficients_a_k} together with~\eqref{eq:norm_of_mu_squared}, there is a global minimum of the function $\widetilde{g}$.
	
	%	To conclude, the minimum of the function $g$ on $B_\kappa(0)$ is attained by $\rho^* = \sum_{i=1}^d a_iv_i$, where
	%	\begin{equation}\label{eq:coefficients_worst_case_parameter}
	%		a_i = -\biggl( \frac{\lambda_i}{1-\gamma}+\frac{h}{\psi(\kappa)\lVert \tau^{-1}\ones \rVert} \biggr)^{-1}\biggl\langle h\tau^\transp c+\frac{\lambda_i}{1-\gamma}\tau^{-1}\nu, v_i\biggr\rangle
	%	\end{equation}
	%	for $i=1,\dots,d$, and where $\psi(\kappa)=-a_1\in(0,\kappa]$ is uniquely determined by $\lVert\rho^*\rVert=\kappa$.
	%	Note that~\eqref{eq:coefficients_worst_case_parameter} also holds for $i=1$ since $\lambda_1=0$ and $c^\transp\tau v_1=\frac{1}{\lVert \tau^{-1}\ones \rVert}$ by~\eqref{eq:c_transposed_v_1}.
\end{proof}

\begin{proof}[Proof of Theorem~\ref{thm:solution_of_the_inf_sup_problem}]
	For any fixed parameter $\mu\in\R^d$, Proposition~\ref{prop:optimal_strategy_non-robust} gives the optimal strategy for the optimization problem
	\[ \sup_{\pi\in\calA_h(x_0)} \E_\mu\bigl[U_\gamma(X^\pi_T)\bigr]. \]
	With the help of Corollary~\ref{cor:optimal_utility_non-robust} we have seen that minimizing the above expression in $\mu$ on the set $K=\bigl\{ \mu\in\R^d \,\big|\, (\mu-\nu)^\transp \Gamma^{-1}(\mu-\nu) \leq \kappa^2 \bigr\}$ is equivalent to minimizing the function $g\colon B_\kappa(0)\to\R$ from Lemma~\ref{lem:worst-case_parameter} in $\rho$ and then setting $\mu=\nu+\tau\rho$.
	The claim now follows from Lemma~\ref{lem:worst-case_parameter} together with the representation in Proposition~\ref{prop:optimal_strategy_non-robust}.
\end{proof}

\begin{proof}[Proof of Lemma~\ref{lem:representation_of_pi_star}]
	Throughout the proof, let
	\[ a_i = -\biggl( \frac{\lambda_i}{1-\gamma}+\frac{h}{\psi(\kappa)\lVert \tau^{-1}\ones \rVert} \biggr)^{-1}\biggl\langle h\tau^\transp c+\frac{\lambda_i}{1-\gamma}\tau^{-1}\nu, v_i\biggr\rangle \]
	for $i=1,\dots,d$, so that $\tau^{-1}(\mu^*-\nu) = \sum_{i=1}^d a_iv_i$. Due to the form of the $a_i$ we can write
	\[ \sum_{i=1}^d \Bigl(\frac{\lambda_i}{1-\gamma}+\frac{h}{\psi(\kappa)\lVert \tau^{-1}\ones \rVert}\Bigr) a_iv_i = -\sum_{i=1}^d \Bigl\langle h\tau^\transp c+\frac{\lambda_i}{1-\gamma}\tau^{-1}\nu,v_i\Bigr\rangle v_i. \]
	Since the vectors $v_1,\dots,v_d$ form an orthonormal basis of $\R^d$ and are eigenvectors to the eigenvalues $\lambda_1,\dots,\lambda_d$ of the symmetric matrix $\tau^\transp A\tau$, the right-hand side equals
	\begin{equation*}
		\begin{aligned}
			-h\tau^\transp c-\frac{1}{1-\gamma}\sum_{i=1}^d \langle \tau^{-1}\nu,\lambda_i v_i\rangle v_i &= -h\tau^\transp c-\frac{1}{1-\gamma}\sum_{i=1}^d \langle \tau^{-1}\nu,\tau^\transp A\tau v_i\rangle v_i \\
			&= -h\tau^\transp c-\frac{1}{1-\gamma}\sum_{i=1}^d \langle \tau^\transp A\nu, v_i\rangle v_i \\
			&= -h\tau^\transp c-\frac{1}{1-\gamma}\tau^\transp A\nu.
		\end{aligned}
	\end{equation*}
	On the other hand, we get
	\begin{equation*}
		\begin{aligned}
			\sum_{i=1}^d \Bigl(\frac{\lambda_i}{1-\gamma}+\frac{h}{\psi(\kappa)\lVert \tau^{-1}\ones \rVert}\Bigr) a_iv_i
			&= \frac{1}{1-\gamma}\sum_{i=1}^d a_i\lambda_iv_i + \frac{h}{\psi(\kappa)\lVert \tau^{-1}\ones \rVert}\sum_{i=1}^d a_iv_i \\
			&= \frac{1}{1-\gamma}\sum_{i=1}^d a_i\tau^\transp A\tau v_i + \frac{h}{\psi(\kappa)\lVert \tau^{-1}\ones \rVert}\tau^{-1}(\mu^*-\nu) \\
			&= \frac{1}{1-\gamma}\tau^\transp A(\mu^*-\nu) + \frac{h}{\psi(\kappa)\lVert \tau^{-1}\ones \rVert}\tau^{-1}(\mu^*-\nu).
		\end{aligned}
	\end{equation*}
	We have used here that $v_i$ is an eigenvector of $\tau^\transp A\tau$ to the eigenvalue $\lambda_i$ for each $i=1,\dots,d$. In conclusion,
	\begin{equation*}
		\frac{1}{1-\gamma}\tau^\transp A\mu^* = -\frac{h}{\psi(\kappa)\lVert \tau^{-1}\ones \rVert}\tau^{-1}(\mu^*-\nu)-h\tau^\transp c.
	\end{equation*}
	Hence, by using the representation of $\pi^*$ from Theorem~\ref{thm:solution_of_the_inf_sup_problem} we obtain
	\begin{equation*}
		\pi^*_t = \frac{1}{1-\gamma}A\mu^* +hc = (\tau^\transp)^{-1}\Bigl(\frac{1}{1-\gamma}\tau^\transp A\mu^* +h\tau^\transp c\Bigr) = -\frac{h}{\psi(\kappa)\lVert \tau^{-1}\ones \rVert}\Gamma^{-1}(\mu^*-\nu)
	\end{equation*}
	for all $t\in[0,T]$.
\end{proof}

\begin{proof}[Proof of Theorem~\ref{thm:duality_result}]
	Since $\pi^*$ is a strategy that is constant in time and deterministic, we can rewrite the expected utility of terminal wealth as
	\[ \E_\mu\bigl[U_\gamma(X^{\pi^*}_T)\bigr] =
	\begin{cases}
		\frac{x_0^\gamma}{\gamma}\exp\biggl(\gamma T \Bigl(r+(\pi^*_0)^\transp (\mu-r\ones)-\frac{1}{2}\lVert \sigma^\transp\!\pi^*_0\rVert^2\Bigr)+\frac{1}{2}\gamma^2T\lVert \sigma^\transp\!\pi^*_0\rVert^2\biggr), &\gamma\neq 0,\\
		\log(x_0)+ T \Bigl(r+(\pi^*_0)^\transp (\mu-r\ones)-\frac{1}{2}\lVert \sigma^\transp\pi^*_0\rVert^2\Bigr), & \gamma=0.
	\end{cases} \]
	Obviously, for any $\gamma\in(-\infty,1)$ the parameter $\mu\in K$ that minimizes the expression above is the parameter that minimizes $(\pi^*_0)^\transp \mu$.
	For an arbitrary $\theta\in\R^d$, $\theta\neq 0$, an easy calculation shows that the parameter $\mu\in\R^d$ that minimizes $\theta^\transp\mu$ such that $(\mu-\nu)^\transp\Gamma^{-1}(\mu-\nu)\leq\kappa^2$ has the form
	\begin{equation}\label{eq:worst_case_mu_in_ellipsoid}
		\widetilde{\mu} = \nu-\frac{\kappa}{\sqrt{\theta^\transp\Gamma\theta}}\Gamma\theta.
	\end{equation}
	Hence it is sufficient to show that the parameter $\mu^*$ is equal to $\widetilde{\mu}$ from~\eqref{eq:worst_case_mu_in_ellipsoid} for $\theta=\pi^*_0$.
	Using Lemma~\ref{lem:representation_of_pi_star} we have
	\begin{equation*}
		(\pi^*_0)^\transp\Gamma\pi^*_0 = \frac{h^2}{\psi(\kappa)^2\lVert \tau^{-1}\ones \rVert^2}(\mu^*-\nu)^\transp\Gamma^{-1}(\mu^*-\nu) = \frac{h^2\kappa^2}{\psi(\kappa)^2\lVert \tau^{-1}\ones \rVert^2}
	\end{equation*}
	and
	\begin{equation}\label{eq:sqrt_pi_Gamma_pi}
		\sqrt{(\pi^*_0)^\transp\Gamma\pi^*_0} = \frac{h\kappa}{\psi(\kappa)\lVert \tau^{-1}\ones \rVert}.
	\end{equation}
	When rearranging the representation in Lemma~\ref{lem:representation_of_pi_star} for $\mu^*$ and plugging in~\eqref{eq:sqrt_pi_Gamma_pi} we obtain
	\begin{equation*}
		\mu^* = \nu-\frac{\psi(\kappa)\lVert \tau^{-1}\ones \rVert}{h}\Gamma\pi^*_0 = \nu-\frac{\kappa}{\sqrt{(\pi^*_0)^\transp\Gamma\pi^*_0}}\Gamma\pi^*_0.
	\end{equation*}
	We conclude that $\mu^*$ is the parameter that minimizes $(\pi^*_0)^\transp \mu$ over all $\mu\in K$ and therefore the worst possible parameter for the strategy $\pi^*$.
	
	Now, for an arbitrary parameter $\mu\in K$, let $\pi(\mu)=(\pi_t(\mu))_{t\in[0,T]}$ denote the strategy from $\calA_h(x_0)$ that is optimal, given that the drift parameter is $\mu$. Then we know from Theorem~\ref{thm:solution_of_the_inf_sup_problem} that
	\begin{equation}\label{eq:inf_sup_equality}
		\adjustlimits \inf_{\mu\in K} \sup_{\pi\in\calA_h(x_0)} \E_\mu\bigl[U_\gamma(X^\pi_T)\bigr]
		= \inf_{\mu\in K} \E_\mu\bigl[U_\gamma(X^{\pi(\mu)}_T)\bigr]
		= \E_{\mu^*}\bigl[U_\gamma(X^{\pi^*}_T)\bigr].
	\end{equation}
	On the other hand, the fact that $\mu^*$ is the worst parameter for an investor using strategy $\pi^*$ yields
	\begin{equation}\label{eq:sup_inf_inequality}
		\begin{aligned}
			\E_{\mu^*}\bigl[U_\gamma(X^{\pi^*}_T)\bigr] = \inf_{\mu\in K} \E_\mu\bigl[U_\gamma(X^{\pi^*}_T)\bigr] \leq \adjustlimits \sup_{\pi\in\calA_h(x_0)} \inf_{\mu\in K} \E_\mu\bigl[U_\gamma(X^{\pi}_T)\bigr].
		\end{aligned}
	\end{equation}
	Furthermore, we also have
	\[ \adjustlimits \sup_{\pi\in\calA_h(x_0)} \inf_{\mu\in K} \E_\mu\bigl[U_\gamma(X^{\pi}_T)\bigr] \leq \adjustlimits \inf_{\mu\in K} \sup_{\pi\in\calA_h(x_0)} \E_\mu\bigl[U_\gamma(X^{\pi}_T)\bigr] \]
	since the inequality always holds when interchanging supremum and infimum, see for example Ekeland and Temam~\cite[Ch.~VI, Prop.~1.1]{ekeland_temam_1976}.
	Consequently, the inequality in~\eqref{eq:sup_inf_inequality} is an equality and the claim follows.
\end{proof}

\begin{proof}[Proof of Lemma~\ref{lem:asymptotics_of_a_kappa}]
	By acknowledging the dependence on $\kappa$, we write $a_i(\kappa)$ for the coefficients of $\rho^*=\tau^{-1}(\mu^*-\nu)$.
	We have already seen in the proof of Lemma~\ref{lem:worst-case_parameter} that $a_1(\kappa)=-\psi(\kappa)$.
	Hence, the constraint $\lVert\tau^{-1}(\mu^*-\nu)\rVert=\kappa$ implies
	\begin{equation*}\label{eq:equation_for_a_kappa}
		1=\frac{\lVert\tau^{-1}(\mu^*-\nu)\rVert^2}{\kappa^2}=\sum_{i=1}^d \Bigl(\frac{a_i(\kappa)}{\kappa}\Bigr)^2=\Bigl(\frac{\psi(\kappa)}{\kappa}\Bigr)^2+\sum_{i=2}^d \Bigl(\frac{a_i(\kappa)}{\kappa}\Bigr)^2
	\end{equation*}
	due to orthonormality of $v_1,\dots,v_d$. In the following, we show that the sum in the expression above goes to zero as $\kappa$ goes to infinity. To prove this, take some $i\in\{2,\dots,d\}$. We know that
	\[ \Bigl(\frac{a_i(\kappa)}{\kappa}\Bigr)^2=\frac{1}{\kappa^2}\biggl( \frac{\lambda_i}{1-\gamma}+\frac{h}{\psi(\kappa)\lVert \tau^{-1}\ones \rVert} \biggr)^{-2}\biggl\langle h\tau^\transp c+\frac{\lambda_i}{1-\gamma}\tau^{-1}\nu, v_i\biggr\rangle^2, \]
	where the expression in the inner product does not depend on $\kappa$. For the other factor, recall that $\psi(\kappa)>0$ and $\lambda_i>0$. Hence,
	\[ \frac{\lambda_i}{1-\gamma}+\frac{h}{\psi(\kappa)\lVert \tau^{-1}\ones \rVert}>\frac{\lambda_i}{1-\gamma}>0 \]
	and therefore
	\[ \frac{1}{\kappa^2}\biggl( \frac{\lambda_i}{1-\gamma}+\frac{h}{\psi(\kappa)\lVert \tau^{-1}\ones \rVert} \biggr)^{-2} \leq \frac{1}{\kappa^2}\biggl( \frac{\lambda_i}{1-\gamma} \biggr)^{-2}, \]
	where the upper bound goes to zero as $\kappa$ goes to infinity.
	The claim now follows from the fact that $\psi(\kappa)$ is positive for each $\kappa$.
\end{proof}

\begin{proof}[Proof of Proposition~\ref{prop:asymptotics_of_mu_star}]
	Using the same notation as before, as well as the result from the previous lemma, we can deduce that
	\[ \frac{1}{\kappa}\tau^{-1}\bigl(\mu^*(\kappa)-\nu\bigr)=\frac{a_1(\kappa)}{\kappa}v_1+\sum_{i=2}^d \frac{a_i(\kappa)}{\kappa}v_i=-\frac{\psi(\kappa)}{\kappa}v_1+\sum_{i=2}^d \frac{a_i(\kappa)}{\kappa}v_i \]
	goes to $-v_1$ as $\kappa$ goes to infinity. The second claim follows immediately.
\end{proof}

\begin{proof}[Proof of Theorem~\ref{thm:limit_of_optimal_strategy}]
	Recall that by Lemma~\ref{lem:representation_of_pi_star} we can write
	\[ \pi^*_t(\kappa)=-\frac{h}{\psi(\kappa)\lVert \tau^{-1}\ones \rVert}\Gamma^{-1}\bigl(\mu^*(\kappa)-\nu\bigr)=-\frac{h}{\lVert \tau^{-1}\ones \rVert}\frac{\kappa}{\psi(\kappa)}\frac{1}{\kappa}\Gamma^{-1}\bigl(\mu^*(\kappa)-\nu\bigr) \]
	for any $t\in[0,T]$. We then obtain
	\[ \lim_{\kappa\to\infty} \pi^*_t(\kappa)=\frac{h}{\lVert \tau^{-1}\ones \rVert} (\tau^\transp)^{-1}v_1=\frac{h}{\lVert \tau^{-1}\ones \rVert^2}(\tau\tau^\transp)^{-1}\ones=\frac{h}{\ones^\transp\Gamma^{-1}\ones}\Gamma^{-1}\ones \]
	by combining the results from Lemma~\ref{lem:asymptotics_of_a_kappa} and Proposition~\ref{prop:asymptotics_of_mu_star}.
\end{proof}

\begin{proof}[Proof of Proposition~\ref{prop:comparison_greater_equal_h_logarithm}]
	Let $\pi'\in\calA'_h(x_0)$ with $\lVert\pi'\rVert\leq M$. Then $\pi'$ can be decomposed as $\pi'_t = \pi_t+\varepsilon_t\ones$ for all $t\in[0,T]$, where $\pi=(\pi_t)_{t\in[0,T]}\in\calA_h(x_0)$ and $\varepsilon_t\geq 0$ for all $t\in[0,T]$. For any fixed $\mu\in K(\kappa)$ we rewrite the expected logarithmic utility given strategy $\pi'$ as
	\begin{equation*}
		\begin{aligned}
			\E_\mu\bigl[\log(X^{\pi'}_T)\bigr] %&= \log(x_0)+rT+\E_\mu\biggl[\int_0^T \Bigl((\pi'_t)^\transp(\mu-r\ones)-\frac{1}{2}\lVert\sigma^\transp\pi'_t\rVert^2\Bigr)\rmd t\biggr]\\
			&= \E_\mu\bigl[\log(X^{\pi}_T)\bigr]+\E_\mu\biggl[\int_0^T \varepsilon_t\Bigl(\ones^\transp(\mu-r\ones)-\frac{1}{2}\varepsilon_t\lVert\sigma^\transp\ones\rVert^2-\ones^\transp\sigma\sigma^\transp\pi_t\Bigr)\rmd t\biggr].
		\end{aligned}
	\end{equation*}
	In particular, we have
	\begin{equation}\label{eq:log_utility_decomposition}
		\begin{aligned}
			&\inf_{\mu\in K(\kappa)} \E_\mu\bigl[\log(X^{\pi'}_T)\bigr] \leq \E_{\mu^*}\bigl[\log(X^{\pi'}_T)\bigr]\\
			&= \E_{\mu^*}\bigl[\log(X^{\pi}_T)\bigr]+\E_{\mu^*}\biggl[\int_0^T \varepsilon_t\Bigl(\ones^\transp\bigl(\mu^*(\kappa)-r\ones\bigr)-\frac{1}{2}\varepsilon_t\lVert\sigma^\transp\ones\rVert^2-\ones^\transp\sigma\sigma^\transp\pi_t\Bigr)\rmd t\biggr],
		\end{aligned}
	\end{equation}
	where $\mu^*=\mu^*(\kappa)$ is the worst-case parameter from Theorem~\ref{thm:solution_of_the_inf_sup_problem}. Our assumption $\lVert\pi'\rVert\leq M$ implies that also $\lVert\pi_t\rVert$ is bounded for every $t\in[0,T]$, and so is $\ones^\transp\sigma\sigma^\transp\pi_t$. Hence there exists a $\kappa_M>0$ such that the second summand in~\eqref{eq:log_utility_decomposition} is non-positive for $\kappa\geq \kappa_M$. That is because $\varepsilon_t\geq 0$ for all $t\in[0,T]$ and
	\[ \lim_{\kappa\to\infty} \ones^\transp\mu^*(\kappa) = \ones^\transp\nu-\lim_{\kappa\to\infty}\psi(\kappa)\ones^\transp\tau v_1 = \ones^\transp\nu-\lim_{\kappa\to\infty}\psi(\kappa)\frac{d}{\lVert\tau^{-1}\ones\rVert}=-\infty. \]
	Since $\kappa_M$ depends only on $M$ but not on the strategy $\pi'$ or its decomposition, we can further deduce
	\[ \adjustlimits \sup_{\substack{\pi\in\calA'_h(x_0)\\ \lVert\pi\rVert\leq M}} \inf_{\substack{\mu\in K(\kappa)\\ \phantom{0}}} \E_\mu\bigl[\log(X^\pi_T)\bigr] \leq \sup_{\pi\in\calA_h(x_0)} \E_{\mu^*}\bigl[\log(X^{\pi}_T)\bigr] = \adjustlimits \sup_{\pi\in\calA_h(x_0)} \inf_{\mu\in K(\kappa)} \E_\mu\bigl[\log(X^\pi_T)\bigr] \]
	for all $\kappa\geq\kappa_M$, which completes the proof.
\end{proof}

\begin{proof}[Proof of Lemma~\ref{lem:properties_of_A_and_c}]
	Using the definition of $A$ in Definition~\ref{def:matrix_A_vector_c} we see that
	\[ A\sigma\sigma^\transp A = D^\transp(D\sigma\sigma^\transp D^\transp)^{-1}D\sigma\sigma^\transp D^\transp(D\sigma\sigma^\transp D^\transp)^{-1}D = D^\transp(D\sigma\sigma^\transp D^\transp)^{-1}D =A, \]
	and hence in particular
	\[ c^\transp\sigma\sigma^\transp A=e_d^\transp(I_d-\sigma\sigma^\transp A)\sigma\sigma^\transp A=e_d^\transp(\sigma\sigma^\transp A-\sigma\sigma^\transp A)=0. \]
	Further, due to $A\ones=0$ we also have
	\[ c^\transp\ones=e_d^\transp(I_d-\sigma\sigma^\transp A)\ones=e_d^\transp\ones=1. \qedhere\]
\end{proof}

\begin{proof}[Proof of Proposition~\ref{prop:comparison_greater_equal_h_power}]
	Take an arbitrary strategy $\pi\in\overline{\calA}_h(x_0)$. Then there exists some $h'\geq h$ such that $\pi\in\calA_{h'}(x_0)$ and we know that
	\[ \inf_{\mu\in K(\kappa)} \E_\mu\bigl[U_\gamma(X^\pi_T)\bigr] \leq \inf_{\mu\in K(\kappa)} \E_\mu\bigl[U_\gamma(X^{\pi'}_T)\bigr] = \E_{\mu'}\bigl[U_\gamma(X^{\pi'}_T)\bigr], \]
	where $\mu'=\mu'(\kappa)$ is the minimizer of the function
	\[ \mu \mapsto \frac{1}{2(1-\gamma)}\mu^\transp A\mu+h'c^\transp\mu \]
	on the uncertainty set $K(\kappa)$ and $\pi'=\pi'(\kappa)\equiv \frac{1}{1-\gamma}A\mu'+h'c$. In the following we show that for sufficiently large level of uncertainty
	\begin{equation}\label{eq:power_utility_compare_h'_and_h}
		\E_{\mu'}\bigl[U_\gamma(X^{\pi'}_T)\bigr] \leq \E_{\mu^*}\bigl[U_\gamma(X^{\pi^*}_T)\bigr]
	\end{equation}
	where $\mu^*=\mu^*(\kappa)$ and $\pi^*=\pi^*(\kappa)$ are the worst-case parameter and the optimal strategy for the utility maximization among strategies in $\calA_h(x_0)$. Note that for showing~\eqref{eq:power_utility_compare_h'_and_h} it is sufficient to prove
	\begin{equation}\label{eq:power_utility_compare_h'_sufficient}
		(\pi'_0)^\transp(\mu'-r\ones)-\frac{1-\gamma}{2}\lVert\sigma^\transp\pi'_0\rVert^2 \leq (\pi^*_0)^\transp(\mu^*-r\ones)-\frac{1-\gamma}{2}\lVert\sigma^\transp\pi^*_0\rVert^2.
	\end{equation}
	Using the representation of $\pi'$ we obtain
	\begin{equation*}
		\begin{aligned}
			&(\pi'_0)^\transp(\mu'-r\ones)-\frac{1-\gamma}{2}\lVert\sigma^\transp\pi'_0\rVert^2\\
			&= \frac{1}{1-\gamma}(\mu')^\transp A\mu'+h'c^\transp(\mu'-r\ones)-\frac{1}{2(1-\gamma)}(\mu')^\transp A\mu'-\frac{1-\gamma}{2}(h')^2c^\transp\sigma\sigma^\transp c\\
			&= \frac{1}{2(1-\gamma)}(\mu')^\transp A\mu'+h'c^\transp\mu'-h'r-\frac{1-\gamma}{2}(h')^2c^\transp\sigma\sigma^\transp c.
		\end{aligned}
	\end{equation*}
	Here we have used the identities from Lemma~\ref{lem:properties_of_A_and_c}. An analogous computation can be done for $\pi^*$ and $\mu^*$. We then see that, since $\mu'$ minimizes
	\[ \mu \mapsto \frac{1}{2(1-\gamma)}\mu^\transp A\mu+h'c^\transp\mu \]
	on $K(\kappa)$, in particular it holds
	\begin{equation*}
		\begin{aligned}
			\frac{1}{2(1-\gamma)}(\mu')^\transp A\mu'+h'c^\transp\mu' &\leq \frac{1}{2(1-\gamma)}(\mu^*)^\transp A\mu^*+h'c^\transp\mu^*.%\\
			%&= \frac{1}{2(1-\gamma)}(\mu^*)^\transp A\mu^*+hc^\transp\mu^*+(h'-h)c^\transp\mu^*.
		\end{aligned}
	\end{equation*}
	Using again $c^\transp\ones=1$ it is easy to show that $c^\transp\mu^*=c^\transp\mu^*(\kappa)$ goes to minus infinity as $\kappa$ goes to infinity. Hence we can choose $\kappa'>0$ such that $c^\transp\mu^*\leq 0$ for all $\kappa\geq\kappa'$. Note that $\kappa'$ does not depend on $\pi'$. For all $\kappa\geq\kappa'$ we then have
	\begin{equation*}
		\begin{aligned}
			&(\pi'_0)^\transp(\mu'-r\ones)-\frac{1-\gamma}{2}\lVert\sigma^\transp\pi'_0\rVert^2\\
			&\leq \frac{1}{2(1-\gamma)}(\mu^*)^\transp A\mu^*+hc^\transp\mu^*+(h'-h)c^\transp\mu^*-h'r-\frac{1-\gamma}{2}(h')^2c^\transp\sigma\sigma^\transp c\\
			&\leq \frac{1}{2(1-\gamma)}(\mu^*)^\transp A\mu^*+hc^\transp\mu^*-hr-\frac{1-\gamma}{2}h^2c^\transp\sigma\sigma^\transp c\\
			&= (\pi^*_0)^\transp(\mu^*-r\ones)-\frac{1-\gamma}{2}\lVert\sigma^\transp\pi^*_0\rVert^2,
		\end{aligned}
	\end{equation*}
	which proves~\eqref{eq:power_utility_compare_h'_sufficient} and hence~\eqref{eq:power_utility_compare_h'_and_h}. Since $\kappa'$ was chosen independent of $h'$ or $\pi'$, we deduce in particular
	\[ \adjustlimits \sup_{\pi\in \overline{\calA}_h(x_0)} \inf_{\mu\in K(\kappa)} \E_\mu\bigl[U_\gamma(X^\pi_T)\bigr] \leq \E_{\mu^*}\bigl[U_\gamma(X^{\pi^*}_T)\bigr] = \adjustlimits \sup_{\pi\in\calA_h(x_0)} \inf_{\mu\in K(\kappa)} \E_\mu\bigl[U_\gamma(X^\pi_T)\bigr] \]
	for all $\kappa\geq\kappa'$. The reverse inequality holds trivially.
\end{proof}

\begin{proof}[Proof of Proposition~\ref{prop:oCOA_greater_than_oRDR}]
	Since both $\pi^*$ and $\hat{\pi}$ are constant in time and deterministic, we can show for $\gamma\neq 0$ that
	\begin{equation}\label{eq:oCOA}
		\begin{aligned}
			\COA &= x_0\rme^{rT} \biggl( \exp\Bigl(T\Bigl((\hat{\pi}_0)^\transp(\nu-r\ones)-\frac{1-\gamma}{2}\lVert\sigma^\transp\hat{\pi}_0\rVert^2\Bigr)\Bigr)\\
			&\qquad\qquad\qquad-\exp\Bigl(T\Bigl((\pi^*_0)^\transp(\nu-r\ones)-\frac{1-\gamma}{2}\lVert\sigma^\transp\pi^*_0\rVert^2\Bigr)\Bigr) \biggr)
		\end{aligned}
	\end{equation}
	and
	\begin{equation}\label{eq:oRDR}
		\begin{aligned}
			\RDR &= x_0\rme^{rT} \biggl( \exp\Bigl(T\Bigl((\pi^*_0)^\transp(\mu^*-r\ones)-\frac{1-\gamma}{2}\lVert\sigma^\transp\pi^*_0\rVert^2\Bigr)\Bigr)\\
			&\qquad\qquad\qquad-\exp\Bigl(T\Bigl((\hat{\pi}_0)^\transp(\mu^*-r\ones)-\frac{1-\gamma}{2}\lVert\sigma^\transp\hat{\pi}_0\rVert^2\Bigr)\Bigr) \biggr).
		\end{aligned}
	\end{equation}
	For $\gamma=0$ we obtain the same representations as in~\eqref{eq:oCOA} and~\eqref{eq:oRDR} with $\gamma=0$. We now plug in the representations from~\eqref{eq:representation_of_pi_hat_robustness}, respectively~\eqref{eq:representation_of_pi_star_robustness}, of the strategies $\pi^*$ and $\hat{\pi}$ and use the properties $A\ones=0$, $c^\transp\sigma\sigma^\transp A=0$ and $A\sigma\sigma^\transp A=A$, see Lemma~\ref{lem:properties_of_A_and_c}. We obtain
	\begin{equation*}
		\begin{aligned}
			\frac{\COA}{x_0\rme^{rT}} 
			&= \overline{L}(\gamma,\kappa)\exp\Bigl(T\Bigl(-hr-\frac{1-\gamma}{2}h^2c^\transp\sigma\sigma^\transp c+hc^\transp\nu+\frac{1}{2(1-\gamma)}\nu^\transp A\nu\Bigr)\Bigr),\\
			\frac{\RDR}{x_0\rme^{rT}} &= \overline{L}(\gamma,\kappa)\exp\Bigl(T\Bigl(-hr-\frac{1-\gamma}{2}h^2c^\transp\sigma\sigma^\transp c+hc^\transp\mu^*+\frac{1}{2(1-\gamma)}(\mu^*)^\transp A\mu^*\Bigr)\Bigr),
		\end{aligned}
	\end{equation*}
	where
	\[ \overline{L}(\gamma,\kappa) = 1-\exp\Bigl(-\frac{T}{2(1-\gamma)}(\mu^*-\nu)^\transp A(\mu^*-\nu)\Bigr). \]
	Hence, we can deduce in particular that
	\[ \frac{\COA}{\RDR}=\frac{\exp\Bigl(T\Bigl(\frac{1}{2(1-\gamma)}\nu^\transp A\nu+hc^\transp\nu\Bigr)\Bigr)}{\exp\Bigl(T\Bigl(\frac{1}{2(1-\gamma)}(\mu^*)^\transp A\mu^*+hc^\transp\mu^*\Bigr)\Bigr)}\geq 1, \]
	since $\mu^*$ minimizes the function $\mu\mapsto \frac{1}{2(1-\gamma)}\mu^\transp A\mu+hc^\transp\mu$ on the set $K$.
\end{proof}

\begin{proof}[Proof of Proposition~\ref{prop:limit_of_oCOA_and_oRDR}]
	Firstly, note that by the same reasoning as in the proof of Theorem~\ref{thm:duality_result} we have
	\[ (\hat{\pi}_0)^\transp\mu^*\leq(\pi^*_0)^\transp\mu^*=(\pi^*_0)^\transp\nu-\kappa\sqrt{(\pi^*_0)^\transp\Gamma\pi^*_0}, \]
	and that the right-hand side goes to $-\infty$ as $\kappa$ goes to infinity. It follows that
	\[ \lim_{\kappa\to\infty}\E_{\mu^*}\bigl[U_\gamma(X^{\hat{\pi}}_T)\bigr]=\lim_{\kappa\to\infty}\E_{\mu^*}\bigl[U_\gamma(X^{\pi^*}_T)\bigr]=
	\begin{cases}
		-\infty,&\gamma\leq 0,\\
		0,&\gamma>0,
	\end{cases} \]
	and therefore $\lim_{\kappa\to\infty}\RDR(\kappa)=0$.
	For $\COA$ we observe that $\E_\nu[U_\gamma(X^{\pi^*}_T)]$ converges to a finite value as $\kappa$ goes to infinity, with that limit being different from zero if $\gamma\neq 0$. It follows that $U_\gamma^{-1}(\E_\nu[U_\gamma(X^{\pi^*}_T)])$ also converges. We thus deduce convergence of $\COA(\kappa)$. Since $\COA(\kappa)\geq 0$ for any $\kappa$, we know that the limit is non-negative.
\end{proof}

\section*{Acknowledgments}

The authors thank two anonymous referees for helpful comments and suggestions that improved this paper.

%\bibliographystyle{../dissertation_style}
%\bibliography{../bibliography-database}

\begin{thebibliography}{10}
	\expandafter\ifx\csname urlstyle\endcsname\relax
	\providecommand{\doi}[1]{doi:\discretionary{}{}{}#1}\else
	\providecommand{\doi}{doi:\discretionary{}{}{}\begingroup
		\urlstyle{rm}\Url}\fi
	
	\bibitem{analui_2014}
	\textsc{B.~Analui}, \emph{Multistage Stochastic Optimization of Energy
		Portfolios under Model Ambiguity}, Ph.D. thesis, Universit\"{a}t Wien (2014).
	
	\bibitem{biagini_pinar_2017}
	\textsc{S.~Biagini \& M.~{\c{C}}. P{\i}nar}, The robust {Merton} problem of an
	ambiguity averse investor, \emph{Mathematics and Financial Economics}
	\textbf{11} (2017), no.~1, pp. 1--24.
	
	\bibitem{bjoerk_khapko_murgoci_2017}
	\textsc{T.~Bj{\"o}rk, M.~Khapko \& A.~Murgoci}, On time-inconsistent stochastic
	control in continuous time, \emph{Finance and Stochastics} \textbf{21}
	(2017), pp. 331--360.
	
	\bibitem{chen_epstein_2002}
	\textsc{Z.~Chen \& L.~Epstein}, Ambiguity, risk, and asset returns in
	continuous time, \emph{Econometrica} \textbf{70} (2002), no.~4, pp.
	1403--1443.
	
	\bibitem{delage_kuhn_wiesemann_2019}
	\textsc{E.~Delage, D.~Kuhn \& W.~Wiesemann}, {``Dice''}-sion--making under
	uncertainty: When can a random decision reduce risk?, \emph{Management
		Science} \textbf{65} (2019), no.~7, pp. 3282--3301.
	
	\bibitem{demiguel_garlappi_nogales_uppal_2009}
	\textsc{V.~DeMiguel, L.~Garlappi, F.~J. Nogales \& R.~Uppal}, A generalized
	approach to portfolio optimization: improving performance by constraining
	portfolio norms, \emph{Management Science} \textbf{55} (2009), no.~5, pp.
	798--812.
	
	\bibitem{ekeland_temam_1976}
	\textsc{I.~Ekeland \& R.~Temam}, \emph{Convex Analysis and Variational
		Problems}, North-Holland Publishing Company (1976).
	
	\bibitem{garlappi_uppal_wang_2007}
	\textsc{L.~Garlappi, R.~Uppal \& T.~Wang}, Portfolio selection with parameter
	and model uncertainty: A multi-prior approach, \emph{The Review of Financial
		Studies} \textbf{20} (2007), no.~1, pp. 41--81.
	
	\bibitem{gilboa_schmeidler_1989}
	\textsc{I.~Gilboa \& D.~Schmeidler}, Maxmin expected utility with non-unique
	prior, \emph{Journal of Mathematical Economics} \textbf{18} (1989), no.~2,
	pp. 141--153.
	
	\bibitem{knight_1921}
	\textsc{F.~H. Knight}, \emph{Risk, Uncertainty and Profit}, Houghton Mifflin,
	Boston (1921).
	
	\bibitem{kramkov_schachermayer_1999}
	\textsc{D.~Kramkov \& W.~Schachermayer}, The asymptotic elasticity of utility
	functions and optimal investment in incomplete markets, \emph{The Annals of
		Applied Probability} \textbf{9} (1999), no.~3, pp. 904--950.
	
	\bibitem{kramkov_schachermayer_2003}
	\textsc{D.~Kramkov \& W.~Schachermayer}, Necessary and sufficient conditions in
	the problem of optimal investment in incomplete markets, \emph{The Annals of
		Applied Probability} \textbf{13} (2003), no.~4, pp. 1504--1516.
	
	\bibitem{lin_riedel_2014}
	\textsc{Q.~Lin \& F.~Riedel}, Optimal consumption and portfolio choice with
	ambiguity (2014). \href{https://arxiv.org/abs/1401.1639}{arXiv:1401.1639}
	[q-fin.PM].
	
	\bibitem{lin_riedel_2021}
	\textsc{Q.~Lin \& F.~Riedel}, Optimal consumption and portfolio choice with
	ambiguous interest rates and volatility, \emph{Economic Theory} \textbf{71}
	(2021), pp. 1189--1202.
	
	\bibitem{maccheroni_marinacci_rustichini_2006}
	\textsc{F.~Maccheroni, M.~Marinacci \& A.~Rustichini}, Ambiguity aversion,
	robustness, and the variational representation of preferences,
	\emph{Econometrica} \textbf{74} (2006), no.~6, pp. 1447--1498.
	
	\bibitem{merton_1969}
	\textsc{R.~C. Merton}, Lifetime portfolio selection under uncertainty: the
	continuous-time case, \emph{The Review of Economics and Statistics}
	\textbf{51} (1969), no.~3, pp. 247--257.
	
	\bibitem{mueller_2005}
	\textsc{M.~M\"{u}ller}, \emph{Market Completion and Robust Utility
		Maximization}, Ph.D. thesis, Humboldt-Universit\"{a}t zu Berlin (2005).
	
	\bibitem{neufeld_nutz_2018}
	\textsc{A.~Neufeld \& M.~Nutz}, Robust utility maximization with {L}\'{e}vy
	processes, \emph{Mathematical Finance} \textbf{28} (2018), no.~1, pp.
	82--105.
	
	\bibitem{oksendal_sulem_2008}
	\textsc{B.~\O{}ksendal \& A.~Sulem}, A game theoretic approach to martingale
	measures in incomplete markets, \emph{Surveys of Applied and Industrial
		Mathematics (TVP Publishers, Moscow)} \textbf{15} (2008), pp. 18--24.
	
	\bibitem{oksendal_sulem_2011}
	\textsc{B.~\O{}ksendal \& A.~Sulem}, Robust stochastic control and equivalent
	martingale measures, in \emph{Stochastic Analysis with Financial
		Applications}, vol.~65 of \emph{Progress in Probability}, Springer Basel
	(2011), pp. 179--189.
	
	\bibitem{pflug_pichler_wozabal_2012}
	\textsc{G.~Pflug, A.~Pichler \& D.~Wozabal}, The {$1/N$} investment strategy is
	optimal under high model ambiguity, \emph{Journal of Banking \& Finance}
	\textbf{36} (2012), no.~2, pp. 410--417.
	
	\bibitem{pham_wei_zhou_2018}
	\textsc{H.~Pham, X.~Wei \& C.~Zhou}, Portfolio diversification and model
	uncertainty: a robust dynamic mean-variance approach (2018).
	\href{https://arxiv.org/abs/1809.01464}{arXiv:1809.01464} [q-fin.PM].
	
	\bibitem{quenez_2004}
	\textsc{M.-C. Quenez}, Optimal portfolio in a multiple-priors model, in
	\textsc{R.~C. Dalang, M.~Dozzi \& F.~Russo}, eds., \emph{Seminar on
		Stochastic Analysis, Random Fields and Applications IV}, vol.~58 of
	\emph{Progress in Probability}, Birkh\"{a}user, Basel (2004), pp. 291--321.
	
	\bibitem{sass_westphal_2021}
	\textsc{J.~Sass \& D.~Westphal}, Robust utility maximization in a multivariate
	financial market with stochastic drift, \emph{International Journal of
		Theoretical and Applied Finance} \textbf{24} (2021), no.~4. 28 pages.
	
	\bibitem{schied_2005}
	\textsc{A.~Schied}, Optimal investments for robust utility functionals in
	complete market models, \emph{Mathematics of Operations Research} \textbf{30}
	(2005), no.~3, pp. 750--764.
	
	\bibitem{schied_2007}
	\textsc{A.~Schied}, Optimal investments for risk- and ambiguity-averse
	preferences: a duality approach, \emph{Finance and Stochastics} \textbf{11}
	(2007), no.~1, pp. 107--129.
	
	\bibitem{schmeidler_1989}
	\textsc{D.~Schmeidler}, Subjective probability and expected utility without
	additivity, \emph{Econometrica} \textbf{57} (1989), no.~3, pp. 571--587.
	
	\bibitem{westphal_2019}
	\textsc{D.~Westphal}, \emph{Model Uncertainty and Expert Opinions in
		Continuous-Time Financial Markets}, Ph.D. thesis, Technische Universit\"{a}t
	Kaiserslautern (2019).
	
	\bibitem{zawisza_2018}
	\textsc{D.~Zawisza}, A note on the worst case approach for a market with a
	stochastic interest rate, \emph{Applicationes Mathematicae} \textbf{45}
	(2018), no.~2, pp. 151--160.
	
\end{thebibliography}

\end{document}